\documentclass[11pt,USenglish,thm-restate]{article}

\usepackage{authblk}

\usepackage[margin=0.8in,letterpaper]{geometry}

\usepackage[utf8]{inputenc}

\usepackage{anyfontsize}

\usepackage{amssymb}
\usepackage{amsthm}
\usepackage{amsmath}
\usepackage{amsfonts}
\usepackage{nccmath}
\usepackage{bbm}
\usepackage[b]{esvect}
\usepackage{mathrsfs}
\usepackage{mathtools}
\usepackage{centernot}
\usepackage{xspace}
\usepackage{stmaryrd}

\usepackage{color}
\usepackage{verbatim}

\usepackage{graphicx}

\usepackage{lscape}
\usepackage{float}
\usepackage{supertabular}
\usepackage{multirow}
\usepackage{longtable}
\usepackage{enumerate}
\usepackage[inline]{enumitem}
\usepackage[framemethod=tikz]{mdframed} 
\usepackage{algpseudocode}
\usepackage{orcidlink}

\usepackage{hyperref} 
\hypersetup{
pdftitle={Round and Resilience-Optimal Approximate Agreement on Trees and Block Graphs},
pdfauthor={Marc Fuchs, Diana Ghinea, Zahra Parsaeian, Joel Rybicki},
pdfkeywords={approximate agreement, trees, optimal round complexity}
}

\usepackage[capitalize]{cleveref}

\usepackage{makecell}

\usepackage{thm-restate}

\usepackage{graphicx,wrapfig,lipsum}

\usepackage{caption}
\usepackage{subcaption}

\usepackage{textcomp}
\usepackage{pgfplots}
\usepgfplotslibrary{fillbetween}

\newtheorem{theorem}{Theorem}
\newtheorem{lemma}{Lemma}
\newtheorem{corollary}{Corollary}
\newtheorem{definition}{Definition}

\newtheorem{remark}{Remark}

\makeatletter
\renewcommand{\paragraph}[1]{\vspace{0.25cm} \noindent{\bf #1}}
\makeatother

\RequirePackage{algpseudocode}

\newcommand{\algoHeadNoBold}[1]{ \underline{#1}}

\makeatletter
\algnewcommand{\ExtendedState}[1]{\State
\parbox[t]{\dimexpr\linewidth-\ALG@thistlm}{\hangindent=\algorithmicindent\strut\hangafter=3#1\strut}}
\makeatother
\algnewcommand\algorithmicinput{\textbf{Input:}}
\algnewcommand\Input{\item[\algorithmicinput]}

\algrenewcommand{\algorithmiccomment}[1]{{\color{gray}// #1}}
\algnewcommand{\IIf}[1]{\State\algorithmicif\ #1\ \algorithmicthen}
\algnewcommand{\EndIIf}{\unskip\ \algorithmicend\ \algorithmicif}

 \RequirePackage{varwidth}
 \RequirePackage{color}
 \RequirePackage[most]{tcolorbox}

\newtcolorbox{titlebox}[5]{enhanced,center,colframe=black,colback=white,boxrule={#3},arc={#2},auto outer arc,%
 breakable,pad at break*=5pt,vfill before first,before={},after={},top=12pt,left=4pt,%
 enlarge top by=2pt,
 fontupper=\small,
 title={\rule[-.3\baselineskip]{0pt}{\baselineskip}\normalsize\sffamily\bfseries #1}, varwidth boxed title*=-30pt, 
 attach boxed title to top left={yshift=-10pt,xshift=10pt}, coltitle=black,
 boxed title style={colback=white,boxrule={#5},arc={#4},auto outer arc}
}

 \newenvironment{dianabox}[1]
 {\begin{titlebox}{\normalfont #1}{0.5pt}{0.5pt}{1pt}{0.75pt}}
 {\end{titlebox}}

\makeatletter
\let\orgdescriptionlabel\descriptionlabel
\renewcommand*{\descriptionlabel}[1]{%
  \let\orglabel\label
  \let\label\@gobble
  \phantomsection
  \edef\@currentlabel{#1}%
  \let\label\orglabel
  \orgdescriptionlabel{#1}%
}
\makeatother

\let\emptyset\varnothing

\DeclarePairedDelimiter\abs{\big\lvert}{\big\rvert}

\newcommand{\party}{p\xspace}

\newcommand{\inputt}[0]{\textsc{in}}
\newcommand{\outputt}[0]{\textsc{out}}
\newcommand{\approximateagreement}[0]{\mathcal{AA}}
\newcommand{\convexagreement}[0]{\mathcal{CA}}

\newcommand{\integers}[0]{\mathbb{Z}}

\newcommand{\realvalues}[0]{\mathbb{R}}

\newcommand{\realAA}[0]{\textsf{RealAA}}

\newcommand{\treeAA}[0]{\textsf{TreeAAFromPaths}}

\newcommand{\finalTreeAA}[0]{\textsf{TreeAA}}

\newcommand{\closestInt}[1]{[#1]}

\newcommand{\hull}[1]{\langle{#1}\rangle}

\newcommand{\dfs}{\textsc{dfs}}
\newcommand{\listConstruction}{\textsf{ListConstruction}}

\newcommand{\vertices}[0]{\textsf{V}}
\newcommand{\diameter}[0]{\textsf{D}}
\newcommand{\distance}[0]{\textsf{d}}
\newcommand{\tree}[0]{T}
\newcommand{\graph}[0]{G}
\newcommand{\pathh}[0]{P}
\newcommand{\roott}[0]{v_\textsf{root}}
\newcommand{\lca}[0]{\textsf{lca}}

\newcommand{\projection}[0]{\textsf{proj}}

\newcommand{\roundT}[0]{t}

\newcommand{\listconstruction}[0]{\textsf{ListConstruction}}

\newcommand{\pathfinder}[0]{\textsf{PathsFinder}^0}

\newcommand{\gradecast}[0]{\mathcal{GC}}

\newcommand{\gradecastprotocol}[0]{\textsf{GC}}

\newcommand{\foxpathfinder}[0]{\textsf{PathsFinder}}

\newcommand{\startt}[0]{\textsf{start}}
\newcommand{\last}[0]{\textsf{end}}

\newcommand{\lcp}[0]{\textsf{LCP}}

\newcommand{\treeAAOld}[0]{\textsf{TreeAA}^0}

\newcommand{\pki}[0]{\textsf{PKI}}

\newcommand{\blockgraphAA}[0]{\textsf{BlockGraphAA}}

\usepackage{cite}

\title{Round and Resilience-Optimal Approximate Agreement \\ on Trees and Block Graphs}

\date{}

\author[1]{Marc Fuchs\,\orcidlink{0000-0003-2272-4483}}
\author[2]{Diana Ghinea\,\orcidlink{0000-0002-5294-9459}}
\author[1]{Zahra Parsaeian\,\orcidlink{0009-0006-3848-1796}}
\author[3]{Joel Rybicki\,\orcidlink{0000-0002-6432-6646}}

\affil[1]{University of Freiburg\\
\texttt{\{marc.fuchs,zahra.parsaeian\}@cs.uni-freiburg.de}}
\affil[2]{Lucerne University of Applied Sciences and Arts\\
\texttt{diana.ghinea@hslu.ch}}
\affil[3]{Humboldt University of Berlin\\
\texttt{joel.rybicki@hu-berlin.de}}

\begin{document}
\maketitle

\begin{abstract}
Approximate Agreement ($\approximateagreement$) is a fundamental primitive that, even in the presence of Byzantine faults, allows honest parties to obtain close (but not necessarily identical) outputs that lie within the range of their inputs. While the optimal round complexity of synchronous $\approximateagreement$ on real values is well understood, its extension to other input spaces has remained open, with fundamental questions regarding achievable resilience and round efficiency still unresolved.

In this work, we investigate the optimal round complexity of synchronous $\approximateagreement$ on trees under Byzantine failures. In this setting, parties hold as inputs vertices of a publicly known labeled tree $\tree$ and must output $1$-close vertices lying in the convex hull of the honest inputs. We present a synchronous protocol with round complexity
$O\left(\frac{\log \diameter(\tree)}{\log \log \diameter(\tree)}\right)$,
where $\diameter(\tree)$ denotes the diameter of the input space tree. Our protocol is obtained via a reduction to real-valued $\approximateagreement$, and maintains the resilience guarantees of the underlying real-valued $\approximateagreement$ protocol: $t = c \cdot n$, for any constant $c < 1/3$ in settings with no cryptographic assumptions, and for any constant $c \in [1/3, 1/2)$ assuming digital signatures.
Complementing this result, we extend impossibility results for real-valued $\approximateagreement$ to any graph $\graph$ by proving a lower bound of $
\Omega\left(\frac{\log \diameter(\graph)}{\log \log \diameter(\graph) + \log \frac{n+t}{t}}\right)$ rounds, where $n$ is the number of parties and $t$ the number of Byzantine faults. Together, these results establish the asymptotic optimality of our protocol whenever $t \in \Theta(n)$.

We further extend our techniques to block graphs by leveraging their clique tree structure. This yields protocols for $\approximateagreement$ on block graphs with optimal resilience in both the synchronous and asynchronous models; in the synchronous setting, the resulting protocols achieve asymptotically optimal round complexity. In doing so, we make progress on a broader open problem: beyond trees, the optimal resilience thresholds for $\approximateagreement$ on graph input spaces were largely unresolved, and our results pin down these thresholds for the class of block graphs. For the exact-agreement variant of the problem, the resilience threshold has been shown to be $t < n / \omega$ on any graph with clique number $\omega$, where $t$ denotes the number of Byzantine faults. Our work shows that a similar lower bound does not hold in general for $\approximateagreement$.

Finally, in contrast to prior protocols for $\approximateagreement$ beyond the real line, which typically rely on safe-area constructions and can incur exponential local computation, our approach in the synchronous model reduces the problem to real-valued $\approximateagreement$ and therefore admits efficient local computation.
\end{abstract}

\thispagestyle{empty}

\newpage
\pagenumbering{arabic}

\section{Introduction}

Ensuring consistency among parties in a distributed system is essential, yet it is a difficult task in the face of potential failures or malicious behavior. Agreement protocols serve as indispensable tools for achieving consensus in such environments. One such fundamental primitive is Approximate Agreement ($\approximateagreement$) \cite{JACM:DLPSW86}. The Byzantine variant of this problem considers a setting of $n$ parties in a network, where each party holds a real value as input. Even when $t$ out of the $n$ parties involved are Byzantine (i.e., malicious), $\approximateagreement$ enables the honest parties to obtain \emph{$\varepsilon$-close} values (for any predefined error $\varepsilon > 0$) that lie within the range spanned by the honest inputs. This relaxed form of agreement has proven useful in scenarios where exact consensus is either unnecessary or infeasible, such as clock synchronization \cite{PODC:LenLos22}, blockchain oracles \cite{delphi24}, distributed machine learning \cite{SuVai16,federated20, federated21}, aviation control systems \cite{OPODIS:StolWat15, IEEE:MelWat18}, and robot gathering \cite{NBiS:PoRaTi11}.

The $\approximateagreement$ problem is not limited to real-valued inputs. Several variants have been explored, including multidimensional real inputs \cite{PODC:VaiGar13, STOC:MenHer13, DIST:MHVG15} and discrete domains such as various classes of (connected) graphs \cite{DISC:NoRy19, eprint:ConvexWorld, OPODIS:Liu23, DISC:ACFR19, OPODIS:AttWel24}. These generalizations naturally extend the original requirements: the honest outputs still need to be close and they must be in the \emph{convex hull} of the honest inputs (defined according to the convexity notion associated with the input space). Intuitively, this means that honest outputs must lie ``inside the region” spanned by the honest
inputs; formal definitions will be provided later.

Regardless of the input space considered, the $\approximateagreement$ problem admits solutions that follow a common iteration-based outline.
In each iteration, the parties use a mechanism to distribute their current values to all parties.
Based on the values received, each party computes a \emph{safe area}: this is a set guaranteed to lie within the convex hull of the values distributed by the honest parties. The honest parties' safe areas are not necessarily identical, but they have some overlap.  Each party then computes its new value essentially as a midpoint of the safe area it has obtained. The parties' new values remain within the honest inputs' convex hull, and they become closer with each iteration. This way, after a sufficient number of iterations, $\approximateagreement$ is achieved.

While this outline has been a powerful tool in establishing sufficient conditions for various input spaces, multiple questions remain open: both in terms of resilience and efficiency (e.g., rounds, communication, local computation). Except for the communication complexity of short inputs, $\approximateagreement$ on real values is well understood in the synchronous model \cite{SokDiana, Fekete90, FeketeFixed87, BenDoHo10, PODC:GhLiWa25, MoseArxivNew}. Moreover, while optimal resilience thresholds are known for spaces such as $\mathbb{R}^d$ \cite{DIST:MHVG15} and trees \cite{DISC:NoRy19, eprint:ConvexWorld}, a unified understanding of optimal efficiency tradeoffs, as well as resilience in more general graph classes, remains elusive.

This paper makes a step forward by exploring optimal round complexity for $\approximateagreement$ under Byzantine faults beyond the real line, namely on trees, in the synchronous model. In this variant~\cite{DISC:NoRy19}, the input space is a publicly-known labeled tree $\tree$. 
Each party holds a vertex of $\tree$ as input, and every honest party must output a vertex of $\tree$ such that the following conditions hold:
(i) honest outputs must be within distance one of each other, and (ii) honest outputs must lie in the honest inputs' convex hull (in the smallest subtree of $\tree$ containing all honest input vertices).

\subsection{Our Contribution}
\paragraph{Results for trees.}
We establish tight asymptotic round-complexity bounds for synchronous $\approximateagreement$ on trees under Byzantine faults.
We first show that 
$\Omega\!\left( \frac{\log \diameter(\graph)}{\log \log \diameter(\graph) + \log\frac{n + t}{t}} \right)$ rounds are necessary for $\approximateagreement$ on any graph $\graph$, where $\diameter(\graph)$ denotes the diameter of $\graph$. This result is obtained by adapting Fekete's bound \cite{Fekete90} from $\approximateagreement$ on real numbers. Afterwards, we show that this lower bound is tight for $\approximateagreement$ on a tree 
whenever $t \in \Theta(n)$ by reducing the problem to real-valued $\approximateagreement$: we achieve $\approximateagreement$ on any tree $\tree$ with round complexity of $O\!\left(\frac{\log \diameter(\tree)}{\log\log \diameter(\tree)}\right)$ by using a round-optimal real-valued $\approximateagreement$ protocol as a building block. At the same time, our protocol achieves optimal resilience: depending on the real-valued $\approximateagreement$ protocol used, our protocol is resilient up to $t < n/3$ Byzantine parties without cryptographic assumptions, and up to $t = c \cdot n$ Byzantine parties, for any constant $c \in [1/3, 1/2)$, given a public-key infrastructure and a secure digital signature scheme.

\paragraph{Extension beyond trees: block graphs.} We show that our techniques extend also to a more general class of graphs, namely block graphs. Block graphs admit a clique tree in which adjacent cliques intersect in exactly one vertex, which enables us to reduce $\approximateagreement$ on any connected block graph $\graph$ to $\approximateagreement$ on a tree $\tree$. Such clique tree-based reductions have appeared before in the wait-free setting (i.e., asynchronous shared memory, with up to $t < n$ crashes) with weaker validity conditions \cite{DISC:ACFR19}.
Using our protocol for trees, we get a synchronous protocol
with asymptotically optimal round complexity $O\left(\frac{\log \diameter{(\graph)}}{\log\log \diameter(\graph)} \right)$ whenever $t \in \Theta(n)$. 
The resilience guarantees are inherited from the underlying real-valued $\approximateagreement$ protocol:  $t < n/3$ Byzantine parties without cryptographic assumptions, and up to $t = c \cdot n$ Byzantine parties, where $c \in [1/3, 1/2)$ is a constant, given a public-key infrastructure and a secure digital signature scheme. Moreover, if we instead use the asynchronous protocol for trees of \cite{DISC:NoRy19} instead of our synchronous protocol, we get an asynchronous protocol for $\graph$ with optimal resilience $t < n/3$.

In comparison to prior work, we observe the following:
\begin{itemize}
    \item Our results improve over previously known sufficient conditions for achieving $\approximateagreement$ on block graphs of maximum clique size $\omega$: $t <n / \omega$ in the synchronous model, and $t < n / (\omega + 1)$ in the asynchronous model \cite{DISC:NoRy19, eprint:ConvexWorld}. These conditions were only known to be tight for the exact-agreement variant of the problem, known as Convex Agreement ($\convexagreement$). As $\convexagreement$ and $\approximateagreement$ have matching optimal resilience thresholds for $\mathbb{R}, \mathbb{R}^d$ and trees, it may be natural to assume that the optimal resilience thresholds match for other input spaces as well. Our result for block graphs shows that this is not the case, and leaves the following question open: characterizing the boundary at which $\approximateagreement$ and $\convexagreement$ separate in resilience.
    
    \item Previous protocols achieving $\approximateagreement$ on spaces beyond the real line involve high local computation complexity due to the safe-area approach, which requires each party to intersect up to $\binom{n}{n - t}$ convex sets in each round \cite{STOC:MenHer13, PODC:VaiGar13,DIST:MHVG15, SPAA:GhLiWa23, DISC:NoRy19, eprint:ConvexWorld}. Our synchronous protocol for block graphs avoids this overhead -- without relaxing the requirements on the honest outputs, by reducing the problem to real-valued $\approximateagreement$, the per-round local computation comes down to sorting up to $n$ real values, discarding the outliers (i.e., the lowest $t$ and the highest $t$) and computing the average of the remaining values. It would be interesting to see if the safe area computation can be avoided for other input spaces.
\end{itemize}

\subsection{Related work}
\paragraph{Real values.}
$\approximateagreement$ was first introduced by Dolev et al.~\cite{JACM:DLPSW86}. They investigated the feasibility of achieving $\approximateagreement$ on $\realvalues$, both in the synchronous and the asynchronous model. Concretely, Dolev et al.~\cite{JACM:DLPSW86} proposed a synchronous protocol resilient against $t < n / 3$ corruptions, which is optimal in settings with no cryptographic assumptions, and an asynchronous protocol resilient against $t < n/5$ corruptions. The resilience condition $t < n / 3$ was proven to be optimal in the asynchronous model as well as in \cite{Coan88, OPODIS:AAD04}, and the condition $t < n/2$ was proven to be tight in the synchronous model with cryptographic assumptions \cite{PODC:GhLiWa22, PODC:LenLos22}. 
Real-valued $\approximateagreement$ has also been completely characterized in terms of resilience in the \emph{network-agnostic model} in \cite{PODC:GhLiWa22}: a network-agnostic protocol must tolerate up to $t_s$ corruptions if it runs in a synchronous network, and up to $t_a \leq t_s$ corruptions if it runs in an asynchronous network (without knowing the type of network a priori).

These protocols follow the aforementioned iteration-based outline and, given that the honest inputs are $\diameter$-close, incur round complexity of roughly $O(\log (\diameter / \varepsilon))$. The work of \cite{JACM:DLPSW86} showed that this is optimal for a specific class of protocols that do not leverage information obtained in prior iterations. Fekete \cite{Fekete90,FeketeFixed87} investigated this further and provided lower bounds on \emph{how close the honest values may get} after $R$ rounds, considering various types of failures (crash, omission, Byzantine) in both synchronous and asynchronous models. In the synchronous model with Byzantine failures, \cite{Fekete90} shows that any $R$-round protocol has an execution where the range of honest values is reduced by a (multiplicative) factor of at least $\frac{t^R}{R^R \cdot (n + t)^R}$ -- in contrast to the $2^{-R}$ factor obtained by prior $R$-iteration $\approximateagreement$ protocols. This lower bound is asymptotically tight: \cite{Fekete90} has also provided an asymptotically matching protocol assuming $t < n / 4$, with message complexity $O(n^R)$. Later, \cite{BenDoHo10} improved the resilience threshold to $t < n / 3$ and reduced the message complexity to $O(R \cdot n^3)$. In addition, with minor adjustments, the \emph{Proxcensus} protocol of \cite{EUROCRYPT:GhGoLi22} becomes an $\approximateagreement$ protocol that matches Fekete's lower bound for $t = c \cdot n$ corruptions, for any constant $c < 1/2$, assuming digital signatures. For the asynchronous model, the lower bound of \cite{FeketeFixed87} holds for protocols in canonical-round form -- as discussed in \cite{ARXIV:AtFlWe25}, this does not cover protocols such as \cite{OPODIS:AAD04, Coan88}.

Optimizations regarding communication complexity have also been a topic of interest.
The work of~\cite{MoseArxivNew} achieves optimal-resilience asynchronous $\approximateagreement$ with $O(n^2)$ messages per iteration, in contrast to the $O(n^3)$ messages required by earlier solutions~\cite{OPODIS:AAD04}.
In addition,~\cite{PODC:GhLiWa24, PODC:GhLiWa25} show that $O(\ell n)$ bits of communication are sufficient to achieve a stronger variant of $\convexagreement$ on $\ell$-bit integer inputs, provided that $\ell$ is sufficiently large, and \cite{SokDiana} extended this result to real-valued $\approximateagreement$.

\paragraph{Trees and graphs.} 
The problem of approximately agreeing on vertices in a graph was first considered by Alcántara et al.~\cite{DISC:ACFR19} and Nowak and Rybicki~\cite{DISC:NoRy19}.
Alcántara et al. \cite{DISC:ACFR19} introduced two variants: \emph{Edge-Gathering}, where honest parties output vertices that are either identical or adjacent, and \emph{$\mathit{1}$-Gathering}, where the honest outputs form a clique. When the input space is a tree, these two variants are equivalent. These variants were presented in the wait-free model, framed as relaxations of the \emph{robot gathering} problem \cite{GatheringRing, GatheringBipartite}. It is important to note that these formulations differ from the classical $\approximateagreement$ requirements in that they do not require the outputs to lie within the convex hull of honest inputs (unless the honest inputs already lie on the same edge in the case of Edge-Gathering, or on a clique in the case of $1$-Gathering).

Nowak and Rybicki~\cite{DISC:NoRy19} generalized the $\approximateagreement$ problem to arbitrary convexity spaces in a fully-connected network with Byzantine failures. They extended the problem to settings such as trees and chordal graphs, and proposed protocols that operate in the asynchronous model. In addition, their work provides a characterization of the exact agreement version of $\approximateagreement$ on arbitrary convex spaces in the synchronous model. These problems were later explored in the network-agnostic model by~\cite{eprint:ConvexWorld}.
Regarding round complexity, the asynchronous $\approximateagreement$ protocol of~\cite{DISC:NoRy19} for trees achieves round complexity $O(\log \diameter)$, where $\diameter$ is the diameter of the input space tree. While further work focused on improving the message complexity~\cite{MoseArxivNew}, the round complexity of $O(\log \diameter)$ remains the state of the art in the asynchronous model, as well as in the synchronous model (up to our work).
For chordal graphs, \cite{DISC:NoRy19} proposed two protocols: one presented as a protocol for cycle-free semilattices (i.e., a particular case of chordal graphs) with round complexity $O(\diameter)$, and a more efficient protocol with $O(\log \diameter)$ rounds, where $\diameter$ denotes the input space diameter of the graph. However, the latter protocol used a reduction to tree decompositions of the chordal graph~\cite{DISC:NoRy19} and was shown to be incorrect by~\cite{eprint:ConvexWorld}, which provided a network-agnostic protocol with $O(\diameter)$ rounds. Whether $o(\diameter)$ rounds are, in fact, sufficient for chordal graphs has remained an open problem; our work resolves this in the affirmative for the subclass of block graphs.

A central question regarding $\approximateagreement$ on graphs is characterizing the classes of graphs that admit $\approximateagreement$ protocols. While the protocols of~\cite{DISC:NoRy19,eprint:ConvexWorld} assume that the  graph is chordal, multiple works have explored structural properties required for $\approximateagreement$ on graphs in the wait-free~model~\cite{DISC:ACFR19, SIROCCO:Alistarh21, PODC:Ledent21, OPODIS:Liu23}.

\section{Preliminaries}

We describe the model and key concepts and definitions.

\paragraph{Model and Adversary.}
We consider $n$ parties $\party_1, \party_2, \dots, \party_n$ running a protocol in a fully-connected network where links model authenticated channels. We assume that the network is synchronous: the parties' clocks are synchronized and every message is delivered within a publicly known amount of time $\Delta$.
Our protocols assume a (computationally unbounded) adversary that can permanently corrupt up to $t < n / 3$ parties. The adversary is adaptive -- it may choose which parties to corrupt at any point in the protocol's execution. Corrupted parties become Byzantine: they may deviate arbitrarily (maliciously) from the protocol. Our lower bound, however, holds even against a static adversary that has to choose which parties to corrupt at the beginning of the protocol's execution.

\paragraph{Approximate Agreement on $\realvalues$.}
We recall the definition of $\approximateagreement$ on real values, as presented in~\cite{JACM:DLPSW86}.

\begin{definition} \label{def:aa-real}
Let $\Pi(\varepsilon)$ be an $n$-party protocol in which every party holds a value in $\realvalues$ as input and $\varepsilon$ is a publicly known parameter.
We say that $\Pi(\varepsilon)$ achieves $\approximateagreement$ if the following properties hold for any predefined $\varepsilon > 0$ even when up to $t$ of the $n$ parties are corrupted:
\textbf{(Termination)} Every honest party produces an output in $\realvalues$ and terminates; \textbf{(Validity)} Honest parties' outputs lie within the range of the honest inputs; \textbf{(\boldmath$\varepsilon$-Agreement)} If two honest parties output $v$ and $v'$, then $\abs{v - v'} \leq \varepsilon$.
\end{definition}

We use the term \emph{valid value} to refer to a value within the range of the honest inputs.

\paragraph{Notations for graphs.} 
For a graph $\graph=(V,E)$, we use $\vertices(\graph)$ to denote its set of vertices. We denote the shortest path between two vertices $u, v \in \vertices(\graph)$ by $\pathh(u, v)$: as we are concerned with trees and block graphs, these are unique.
We define the distance $\distance(u, v)$ as the length of $\pathh(u, v)$.
We use $\diameter(\graph)$ to denote the diameter of $\graph$, i.e., the maximum distance between any two vertices in $\vertices(\graph)$.
The \emph{(shortest path) convex hull} $\hull{S}$ of a set of vertices $S \subseteq \vertices(\graph)$ is the union of all shortest paths between any pair of vertices in $S$. Note that for trees, this is equivalent to the set of vertices in the smallest connected subtree containing $S$.

\paragraph{$\approximateagreement$ on a graph $\graph$.}
We now define the problem of $\approximateagreement$ on graph $\graph$ ~\cite{DISC:NoRy19}.
\begin{definition} \label{def:aa-block-graphs}
Let $\graph$ be a fixed labeled graph known to all parties and $\Pi(\graph)$ be an $n$-party protocol in which every party holds a vertex of $\graph$ as input.
We say that $\Pi(\graph)$ achieves $\approximateagreement$ on $\graph$ if the following properties hold even when up to $t$ of the $n$ parties are corrupted: \textbf{(Termination)} Every honest party produces an output in $\vertices(\graph)$ and terminates; \textbf{(Validity)} Honest parties' outputs lie within the honest inputs' convex hull; \textbf{(1-Agreement)} If two honest parties output $v$ and $v'$, then~$\distance(v, v') \leq 1$.
\end{definition}

As in the case of real values, we refer to a vertex in the honest inputs' convex hull as a \emph{valid vertex}. If the input space graph $\graph$ has diameter $\diameter(\graph) \le 1$ (i.e., it is a single clique), the $\approximateagreement$ problem becomes trivial: each party can simply return its own input. Therefore, in the remainder of this work, we focus on graphs with diameter $\diameter(\graph)>1$.
We add that, except for the case where $\graph$ is a tree, this does not apply to the \emph{Edge-Agreement} variant of \cite{DISC:ACFR19}, where the honest parties' outputs need to be on the same edge.

\paragraph{Round complexity.} 
In a synchronous communication round, every party may send messages to the other parties, receive all messages sent to it during that round, and perform local computations. The round complexity of a protocol is the maximum number of such rounds until every honest party terminates, over all admissible executions.

\section{Lower Bound}\label{section:lower-bound}

In this section, we analyze the round complexity required to achieve $\approximateagreement$ on a graph. We begin by revisiting Fekete's lower bound for real-valued $\approximateagreement$ regarding \emph{how close the honest values may get} after a fixed number of rounds. Fekete~\cite{Fekete90} used chain arguments to prove lower bounds for the convergence rate of real-valued $\approximateagreement$ protocols. However, upon closer inspection, his proof only relies on a weaker validity condition than \Cref{def:aa-real} called \emph{Strong Unanimity}: if all honest parties hold input $v$, no honest party outputs $v' \neq v$.

Fekete constructed for the $R$-round full-information protocol a chain of 
$s \leq \frac{(n + t)^R}{\roundT_1 \cdot \roundT_2 \cdot \cdots \cdot \roundT_R}$ views, 
where $\sum_{i=1}^R \roundT_i \le t$, and in the first view, all honest parties start with input $a$, and in the last view, all honest parties start with input $b$. For these views, $\roundT_i$ roughly represents the number of Byzantine parties that deviate from the protocol for the first time in round $i$. For any pair of views that are consecutive in the chain, there exists an execution of the full-information protocol in which two honest parties obtain those views. 

For any protocol satisfying strong unanimity, the first view leads to all honest parties outputting~$a$ and the last view leads to all honest parties outputting~$b$. In the case of real values, this implies the existence of two consecutive views in the chain that yield two honest outputs $v, v' \in \realvalues$ with $\abs{v - v'} \geq \abs{a-b} / s$. However, the proof works as-is on any set $V$ of values with a distance metric~$\distance$.
Following the exact same steps as in Fekete's proof~\cite[Theorem 15]{Fekete90}, 
we get the next lemma.

\begin{lemma}\label{corollary:fekete-bound-alt}
Let $\Pi$ be any deterministic $R$-round protocol with inputs and outputs from the set $V$ that satisfies strong unanimity even when up to $t$ of the $n$ parties involved are Byzantine.
Then for any $a,b \in V$, there exists an execution of $\Pi$ in which two honest parties  output $v$ and $v'$ such that $
 \distance(v, v') \geq \distance(a,b) \cdot \left( \frac{t}{R(n+t)} \right)^R.$
\end{lemma}

The lower bound is an immediate corollary of the above lemma. By Fekete's results~\cite[Theorem 18]{Fekete90}, the same asymptotic lower bound holds for crash faults by replacing $t/(n+t)$ with $t/(2n+3t)$. 

\begin{theorem}
Let $n>t>0$ be fixed. Any deterministic $n$-party protocol that achieves $\approximateagreement$ on a graph $G$ even when up to $t$ of the $n$ parties are Byzantine requires $\Omega\left( \frac{\log \diameter(G)}{ \log \log \diameter(G) + \log \frac{n+t}{t}} \right)$ rounds to terminate. \label{thm:lower-bound}
\end{theorem}
\begin{proof}
Let $\delta := (n+t)/t$. 
Choose $a,b \in V$ such that $\distance(a,b) = \diameter(G)$. 
Then by 1-Agreement and \Cref{corollary:fekete-bound-alt}, any honest parties' output values $v$ and $v'$ satisfy 
\[
1 \ge \distance(v,v') \ge \diameter(G) \cdot \left( \frac{t}{R(n+t)} \right)^R =  \diameter(G) \cdot \left(\frac{1}{R\delta}\right)^R.
\]
By rearranging the above, the claim now follows for any fixed $\delta$, as   
\begin{align*}
\left( R\delta \right)^R \ge \diameter(G) &\iff R \log (R\delta) \ge \log \diameter(G) \\
&\iff R \ge \Omega\left( \frac{\log \diameter(G)}{\log \log \diameter(G) + \log \delta}\right).
\end{align*}
\end{proof}

\section{Round-Optimal Protocol for Trees}

In this section, we give our round-optimal $\approximateagreement$ protocol for labeled trees, asymptotically matching the lower bound of \Cref{section:lower-bound}. We begin with the special case where the input space is a path, where the problem reduces cleanly to real-valued $\approximateagreement$. We then lift this warm-up to general trees by projecting inputs onto suitably chosen paths.

\subsection{Warm-up: Protocol for Paths}\label{section:warmup}
Building towards our solution, we begin by describing a warm-up protocol for the case where the input space is a labeled path~$\pathh$.
We make use of the protocol of~\cite{BenDoHo10}, denoted by $\realAA$, which achieves $\approximateagreement$ on $\realvalues$ with asymptotically optimal round complexity, described by the theorem below.
We add that the round complexity analysis presented in \cite{BenDoHo10} assumes $\varepsilon = 1/n$. We extend their analysis to any $\varepsilon > 0$. The technical details of the proof are included in \Cref{appendix:realvalues-aa}.
\begin{restatable}{theorem}{RealValuesAA}\label{theorem:real-values-aa}
    There is a protocol $\realAA(\varepsilon)$ achieving $\approximateagreement$ on real values even when up to $t < n / 3$ of the $n$ parties involved are Byzantine. If the honest inputs are $\diameter$-close, $\realAA(\varepsilon)$ ensures Termination within $R_{\realAA}(\diameter, \varepsilon) < 7 \cdot \frac{\log_2 (\diameter/\varepsilon)}{\log_2 \log_2 (\diameter/\varepsilon)} + 3$ rounds.
\end{restatable}

To achieve $\approximateagreement$ on a path $\pathh$, the parties denote the $k := \diameter(\pathh)$ vertices in the path $\pathh$ by $(v_1, v_2, \ldots, v_k)$ (where $v_i$, $v_{i + 1}$ are adjacent, and $v_1$ is the endpoint of $\pathh$ with the lowest label in lexicographic order).
A party with input $v_i$ joins 
$\realAA(1)$ with input $i$. Each party obtains a real value $j$ from $\realAA(1)$ and outputs the vertex 
$v_{\closestInt{j}}$, where $\closestInt{j}$ denotes the nearest integer to $j$: if $z \leq j < z + 1$ for $z \in \integers$, $\closestInt{j} := z$ if $j - z < (z + 1) - j$ and $\closestInt{j} := z + 1$ otherwise. The following properties are~easy~to~show.

\begin{remark}\label{remark:combined} \hfill
\begin{enumerate}
 \item If $i_{\min}, i_{\max} \in \integers$ and $i_{\min} \leq j \leq i_{\max}$, then $i_{\min} \leq \closestInt{j} \leq i_{\max}$.
 \item  If $j, j' \in \realvalues$ satisfy $\abs{j - j'} \leq 1$, then $\abs{\closestInt{j} - \closestInt{j'}} \leq 1$. 
\end{enumerate}
\end{remark}
The first part of \Cref{remark:combined}
follows directly from
the definition of $\closestInt{\cdot}$ and ensures that the vertices $v_{\closestInt{j}}$ obtained are valid. For $1$-Agreement, $\realAA(1)$ provides the honest parties with $1$-close real values $j$.
The second part of \Cref{remark:combined}
ensures that the honest parties hold $1$-close integer values  $\closestInt{j}$. Hence, the honest parties' output vertices $v_{\closestInt{j}}$ are $1$-close, and $1$-Agreement holds. Consequently, we have achieved $\approximateagreement$ on the path $\pathh$ within $O\left( \frac{\log \diameter(\pathh)}{\log \log \diameter(\pathh)} \right)$ rounds by \Cref{theorem:real-values-aa}.

\subsection{Moving Towards Trees} \label{section:path-assumed}
As a step towards our protocol for trees, we assume that the parties \emph{know} a path $\pathh$ in the input space tree $\tree$ that intersects the honest inputs' convex hull.
Then, the parties may proceed as follows: each party with input $v_{\inputt} \in \vertices(\tree)$ computes the projection of the vertex $v_{\inputt}$ onto the path $\pathh$, denoted by $\projection_{\pathh}(v_{\inputt})$. This is the vertex in $\pathh$ that has the shortest distance to $v_{\inputt}$, as shown in Figure \ref{figure:path-trick}. 
Note that the honest parties' projections $\projection_{\pathh}(v_{\inputt})$ are in the honest inputs' convex hull, as described by the lemma below. The proof is included in \Cref{appendix:realvalues-remarks}.
\begin{restatable}{lemma}{ObviousHullMaintained}\label{remark:convex-hull-maintained}
    Consider a set of vertices $S \subseteq \vertices(\tree)$ and a path $\pathh$ in $\tree$ such that $\vertices(\pathh) \cap \hull{S} \neq \emptyset$. Then, for any $v \in S$, $\projection_{\pathh}(v) \in \vertices(\pathh) \cap \hull{S}$.
\end{restatable}

\begin{figure}[h]
\centering
\vspace{-2em}
\includegraphics[width=0.65\textwidth]{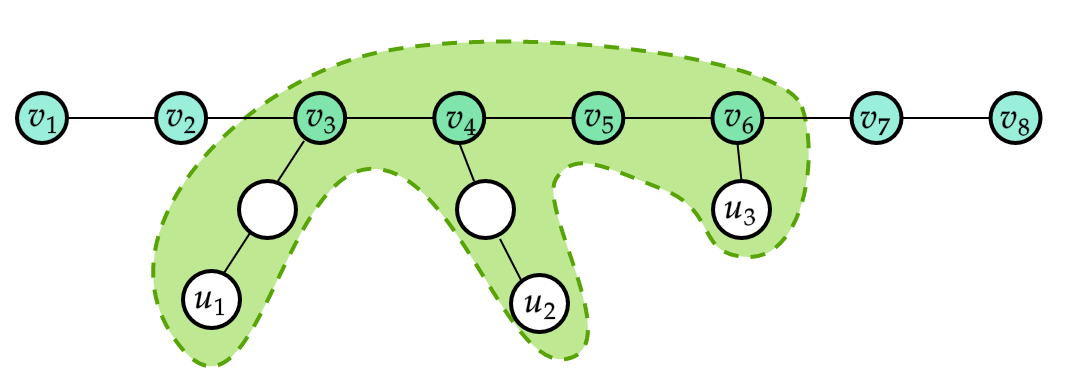}
\caption{Let $\pathh$ be the assumed path, represented by the sequence of vertices $v_1, v_2, \dots, v_8$. The vertices $u_1, u_2, u_3$ correspond to the honest inputs, whose convex hull is highlighted in green. The projections of $u_1, u_2, u_3$ onto path $\pathh$ are vertices $v_3, v_4, v_6$, respectively. 
}\label{figure:path-trick}
\end{figure}
Then, the parties follow the approach described in \Cref{section:warmup} to reach $\approximateagreement$ on the path $\pathh$, using the vertices $\projection_{\pathh}(v_{\inputt})$ as inputs. This provides the parties with vertices that are $1$-close in $\tree$, and by \Cref{remark:convex-hull-maintained}, valid in $\tree$.  Consequently, $\approximateagreement$ is achieved within $O \left( \frac{\log \diameter(\tree)}{\log \log \diameter(\tree)} \right)$ rounds. 

\subsection{Weakening the Path Assumption: Initial Attempt}\label{section:brief-announcement-solution}

In Section \ref{section:path-assumed}, we have discussed a round-optimal approach for $\approximateagreement$ on trees given that the honest parties \emph{know} a path in $\tree$ that intersects the honest inputs' convex hull. This section makes a first step towards removing this assumption, leading to a solution with round complexity $O \left( \frac{\log \abs{\vertices(\tree)}}{\log \log \abs{\vertices(\tree)}} \right)$. This is the protocol we have described in \cite{ourBA}, and represents a stepping stone towards our final solution for trees, presented in \Cref{section:fox-paths}.

Intuitively, it may seem that finding a path $\pathh$ that passes through the honest inputs' convex hull comes down to solving \emph{Byzantine Agreement}. This would require $t + 1 = O(n)$ communication rounds \cite{DolStr83}, which generally prevents us from achieving our round complexity goal.
Instead, we implement a subprotocol $\pathfinder$ that enables the honest parties to \emph{approximately} agree on such a path. Concretely, each honest party obtains a subpath $\pathh$ of $\tree$ such that:
(i) $\pathh$ intersects the convex hull of the honest inputs; and
(ii) if two honest parties obtain different paths $\pathh$ and $\pathh'$, then either $\pathh$ extends $\pathh'$ by one edge, or $\pathh'$ extends $\pathh$ by one edge.
Formally, if $\pathh = (v_1, \ldots, v_k)$ and $\pathh' = (u_1, \ldots, u_{k'})$, then either $\pathh = \pathh'$, or $\pathh' = \pathh \oplus (v_k, u_{k'})$, or $\pathh = \pathh' \oplus (u_{k'}, v_k)$.
In Section~\ref{section:final-protocol}, we show that this suffices to apply the approach from Section~\ref{section:path-assumed}.

\subsubsection{Approximately Agreeing on Paths}\label{section:path}

To obtain these paths, we assume that the input space tree $\tree$ is rooted, and we denote its root by $\roott$ (this will be the vertex in $\tree$ with the lowest label in lexicographic order). We then enable the honest parties to achieve $1$-Agreement on vertices within a subtree rooted at a valid vertex. Once an honest party obtains such a vertex $v$, it is guaranteed that the path $\pathh(\roott, v)$ intersects the convex hull of the honest inputs. Moreover, the $1$-Agreement property ensures that the honest parties' paths $\pathh(\roott, v)$ are identical except possibly for one additional edge.

\paragraph{List representation.} Our approach for enabling honest parties to identify vertices within a subtree rooted at a valid vertex is based on transforming the (now rooted) tree~$\tree$ into a list~$L$ with certain key properties. This transformation relies on a technique commonly used for efficiently computing \emph{lowest common ancestors} in rooted trees~\cite{LCATrick}.
Each party performs a depth-first search ($\dfs$) starting from the fixed root vertex~$\roott$ and records each vertex upon every visit during the $\dfs$ traversal.
We present the code for this step below. Afterwards, we describe the list obtained in this step for the tree depicted in Figure \ref{figure:small-tree}. 
\begin{dianabox}{$\listconstruction(\tree, \roott)$}
\algoHeadNoBold{$\dfs(v)$, where $v \in \vertices(\tree)$, $L$ is a global list:}
	\begin{algorithmic}[1]
            \State Append $v$ to $L$.
            \State For each neighbor $v' \notin L$ of $v$ (ordered by labels):
            \State \hspace{0.5cm} Run $\dfs(v')$. Afterwards, append $v$ to $L$.
    \end{algorithmic}
    \vspace{0.5cm}
    
	\algoHeadNoBold{Code for party $\party$, given tree $\tree$ and a vertex $\roott \in \vertices(\tree)$:}
	\begin{algorithmic}[1]
            \State $L := [\ ]$. Run $\dfs(\roott)$. Afterwards, \Return $L$.
    \end{algorithmic}
\end{dianabox}

\begin{wrapfigure}[7]{r}{5.2cm}
\centering
\vspace{-0.8cm}
\includegraphics[width=4.77cm]{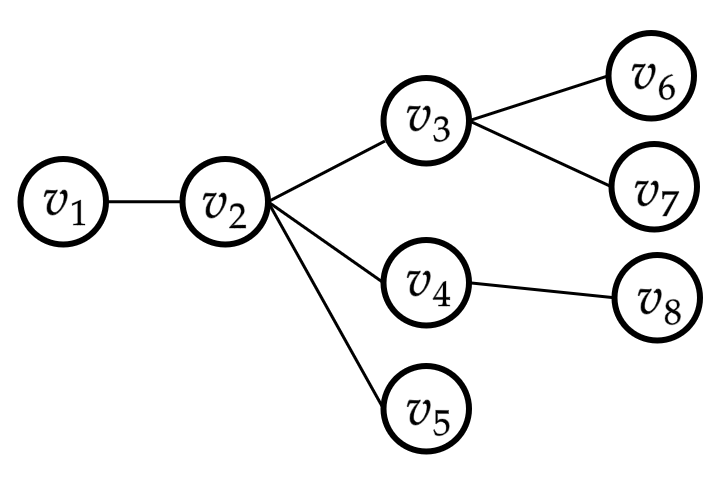}
\vspace{-0.47cm}
\caption{An input space tree.}\label{figure:small-tree}
\end{wrapfigure} 

\Cref{figure:small-tree} illustrates an example of an input space tree. If $\roott := v_1$, every party begins the $\dfs$ from $v_1$. The traversal then visits $v_2$, proceeds to $v_3$, then $v_6$, returns to $v_3$, continues to $v_7$, returns again to $v_3$, then to $v_2$, and so on. The final list obtained from this traversal is $
L = [v_1, v_2, v_3, v_6, v_3, v_7, v_3, v_2, v_4, v_8, v_4, v_2, v_5, v_2, v_1].$

To describe the properties of the resulting list, we first establish a few notations. We use $L_i$ (for $1 \leq i \leq \abs{L}$) to denote the $i$-th element of list~$L$. In addition, for a vertex $v \in \vertices(\tree)$, we let $L(v)$ denote the set of indices $i$ such that $L_i = v$.

Lemma~\ref{lemma:list-construction} formally establishes the guarantees of the list constructed via this approach. The proof is provided in Appendix~\ref{appendix:trees}.

\begin{restatable}{lemma}{ListConstruction} \label{lemma:list-construction}
    Consider a rooted tree $\tree$, and let $\roott$ denote its root.
    
    Then, $\listConstruction(\tree, \roott)$ returns, in finite time, a list of vertices $L$ with the following properties:
    \begin{enumerate}[nosep]
        \item If $\abs{\vertices(\tree)} > 1$, then for any $i < \abs{L}$, the vertices $L_i$ and $L_{i+1}$ are adjacent in $\tree$.
        \item The list $L$ contains $\abs{L} \leq 2 \cdot \abs{\vertices(\tree)}$ elements, and, for every vertex $v \in \vertices(\tree)$, we have $L(v) \neq \emptyset$. 
        \item Consider a vertex $v \in \vertices(\tree)$, and let $i_{\min} = \min L(v)$ and $i_{\max} = \max L(v)$. Then, a vertex $u$ is in the subtree rooted at $v$ if and only if $L(u) \subseteq [i_{\min}, i_{\max}]$.
        \item For any two vertices $v, v' \in \vertices(\tree)$ and any $i \in L(v)$ and $i' \in L(v')$, the lowest common ancestor of $v$ and $v'$ is in the set $\{L_k : \min(i, i') \leq k \leq \max(i, i')\}$.
    \end{enumerate}
\end{restatable}

Note that the $\listConstruction$ algorithm is deterministic, and therefore enables the honest parties to (simultaneously) obtain the same list $L$.

\paragraph{Finding vertices in the subtree of a valid vertex.}
After computing the list $L$, each party $\party$ with input vertex $v_{\inputt}$ joins $\realAA(1)$ with input $i \in L(v_{\inputt})$. 
We add that $L(v_{\inputt})$ may contain multiple indices, and the parties may choose any of them; without loss of generality, we let $i := \min\{ L(v_{\inputt})\}$. 

Then, $\realAA(1)$ provides the parties with $1$-close real values $j$ that lie within the range of honest values $i$. By applying Remark~\ref{remark:combined}, we conclude that the honest parties obtain $1$-close integer values $\closestInt{j}$ that are within the range of honest values $i$. Finally, using \Cref{lemma:list-construction}, it follows that the list elements $L_{\closestInt{j}}$ corresponding to honest parties' values $j$ are $1$-close vertices that belong to a subtree rooted at a valid vertex.

\begin{wrapfigure}{r}{5.3cm}
\centering
\includegraphics[width=5cm]{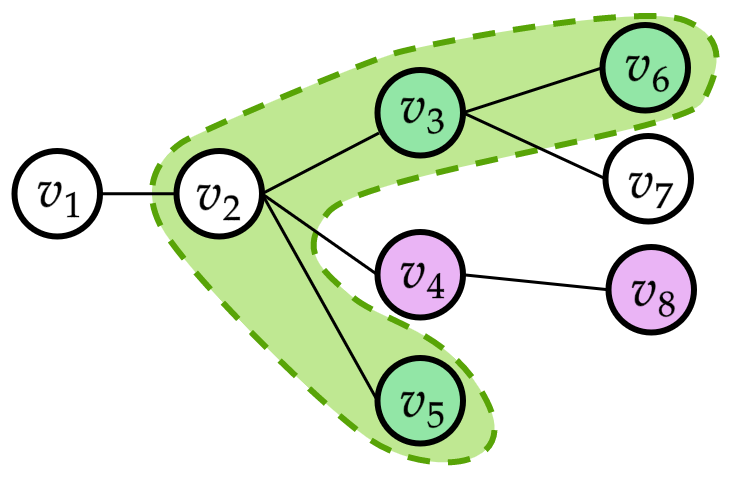}
\caption{Vertices $v_4$, $v_8$ are not valid, but are in the subtree of a valid vertex (with respect to root $v_1$).}\label{figure:small-tree-extra}
\end{wrapfigure} 

We note that this \textbf{does not} imply that the vertices $L_{\closestInt{j}}$ are valid. Consider again the input space tree in Figure~\ref{figure:small-tree}, depicted in Figure~\ref{figure:small-tree-extra}. 
If the honest inputs are $v_3$, $v_6$, and $v_5$, their convex hull is $\{v_5, v_2, v_3, v_6\}$. 
The honest parties join $\realAA(1)$ with indices $i$ chosen from $L(v_3) = \{3, 5, 7\}$, $L(v_6) = \{4\}$, and $L(v_5) = \{13\}$. 
Note that the indices in $L(v_4) = \{9, 11\}$ and $L(v_8) = \{10\}$ lie within the range of honest indices $i$. 
Thus, $\realAA(1)$ may return outputs $j$ such that $\closestInt{j}$ corresponds to, for example, $v_4$ or $v_8$ -- these are vertices that lie outside the honest inputs' convex hull $\{v_5, v_2, v_3, v_6\}$. 
However, $v_4$ and $v_8$ still belong to the subtree rooted at the valid vertex $v_2$.

\paragraph{Paths intersecting the honest inputs' convex hull.}
In the example above, since vertices $v_4$ and $v_8$ lie in the subtree rooted at the valid vertex $v_2$, the paths $\pathh(v_1, v_4)$ and $\pathh(v_1, v_8)$ intersect the convex hull of the honest inputs.
The lemma below shows that this property holds more generally: if an honest party obtains vertex $L_{\closestInt{j}}$ via our approach, then the path $\pathh(\roott, L_{\closestInt{j}})$ intersects the honest inputs' convex hull.

\begin{lemma} \label{lemma:path-crosses-convex-hull}
    Assume a rooted tree $\tree$ with root $\roott$ and let $L := \listconstruction(\tree, \roott)$.
    Consider a non-empty set of vertices $S \subseteq \vertices(\tree)$, and let $i_{\min}$ and $i_{\max}$ denote the lowest and highest indices in $L$ corresponding to vertices in $S$, respectively: $i_{\min} := \min \bigcup_{v \in S} L(v)$ and $i_{\max} := \max \bigcup_{v \in S} L(v)$. 
    
    Then, for any $i$ satisfying $i_{\min} \leq i \leq i_{\max}$, the path $\pathh(\roott, L_i)$ intersects the convex hull of $S$; that is,
    $\vertices(\pathh(\roott, L_i)) \cap \hull{S} \neq \emptyset$.
\end{lemma}
\begin{proof} 
Let $\lca$ be the lowest common ancestor of $L_{i_{\min}}$ and $L_{i_{\max}}$. 
By definition $L_{i_{\min}}, L_{i_{\max}}\in S$, and, as $\lca$ lies on the unique path connecting these two vertices in $\tree$, we have $\lca \in \hull{S}$. In the following, we show that, for each $i$ satisfying $i_{\min} \leq i \leq i_{\max}$, $\lca$ is one of the vertices in path $\pathh(\roott, L_i)$, which implies that  $\pathh(\roott, L_i)$ intersects the convex hull of $S$.

    If $a := \min L(\lca)$ and $b := \max L(\lca)$, \Cref{lemma:list-construction} guarantees that a vertex $v$ is in the subtree rooted at $\lca$ if and only if $L(v) \subseteq [a, b]$. Then, since $L_{i_{\min}}$ and $L_{i_{\max}}$ are in $\lca$’s subtree, we know that $i_{\min}, i_{\max}\in [a,b]$.

 Moreover, by hypothesis $i_{\min}\le i\le i_{\max}$, so $i\in [a,b]$ as well. By the property of $L$ just stated (from \Cref{lemma:list-construction}), it follows that $L_i$ \emph{lies in the subtree rooted at $\lca$}. Equivalently, $\lca$ is an ancestor of $L_i$ in $\tree$. Finally, since $\lca$ is an ancestor of $L_i$, the unique path from $\roott$ to $L_i$ \emph{must} pass through $\lca$.
 Therefore $\vertices(\pathh(\roott, L_i)) \cap \hull{S} \neq \emptyset$, which concludes our proof.
\end{proof}

\paragraph{Subprotocol $\pathfinder$.} We now present our subprotocol $\pathfinder$. Each party computes the list representation $L$ of the (rooted) input space tree $\tree$ using the subprotocol $\listconstruction$, which provides all parties with the same list $L$.

Afterwards, the parties execute $\realAA(1)$ on the list $L$: 
each party joins $\realAA(1)$ with an index in $L$ that corresponds to its input vertex $v_{\inputt}$. 
$\realAA(1)$ then provides the parties with indices corresponding to $1$-close vertices $v$ in $\tree$. Finally, each honest party returns the path from the root vertex $\roott$ to the vertex obtained. 

\begin{dianabox}{$\pathfinder(\tree, \roott, v_{\inputt})$}
	\algoHeadNoBold{Code for party $\party$, given the labeled tree $\tree$ with root vertex $\roott$ and input $v_{\inputt} \in \vertices(\tree)$}
	\begin{algorithmic}[1]
            \State $L := \listConstruction(\tree, \roott)$.
            \State Join $\realAA(1)$ with input $i := \min L(v_{\inputt})$; obtain output $j \in \realvalues$.
            \State 
            \Return $\pathh := \pathh(\roott, L_{\closestInt{j}}).$ 
    \end{algorithmic}
\end{dianabox}

We establish the guarantees of $\pathfinder$ below. 

\begin{lemma} \label{lemma:path-finder}
    Assume a protocol $\realAA$ achieving $\approximateagreement$ on $\mathbb{R}$ with round complexity $R_{\realAA}(\diameter, \varepsilon)$ when the honest inputs are $D$‑close and the target error is $\varepsilon$. Then, if the honest parties join $\pathfinder$ with the same labeled tree $\tree$ and the same root vertex $\roott$, $\pathfinder$ provides each honest party with a subpath $\pathh$ of $\tree$ such that:
    \begin{enumerate}[nosep]
        \item If $S$ denotes the set 
        of honest parties' inputs $v_{\inputt}$,  then $\vertices(\pathh) \cap \hull{S} \neq \emptyset$.
        \item There exists a subpath $\pathh^{\star} = (v_1 = \roott, v_2, \ldots, v_{k^\star + 1})$ of $\tree$ such that, for every honest party, either $\pathh = (v_1, v_2, \ldots, v_{k^\star})$ or $\pathh = (v_1, v_2, \ldots, v_{k^\star + 1})$.
    \end{enumerate}
    Moreover, each honest party obtains its path within $R_{\pathfinder} := R_{\realAA}( 2\cdot\abs{\vertices(\tree)}, 1)$ rounds. 
\end{lemma}
\begin{proof}
We first discuss Property 1.
Lemma \ref{lemma:list-construction} ensures that every vertex of $\tree$ appears in list $L$, and therefore every party obtains a well-defined index $i$.
Afterwards, $\realAA(1)$ ensures that parties obtain $1$-close real values $j$ that lie within the range of honest indices $i$. According to Remark \ref{remark:combined}, the values $\closestInt{j}$ are also within the range of honest indices $i$.
Then, since the honest parties have computed the list $L$ identically, Lemma \ref{lemma:path-crosses-convex-hull} ensures that each honest party obtains a path $\pathh$ that intersects the honest inputs' convex hull; therefore, Property 1 holds.

For Property 2, we note that, since the real values $j$ obtained by the honest parties are $1$-close, Remark \ref{remark:combined} ensures that the integer values $\closestInt{j}$ are also $1$-close. Moreover, Lemma \ref{lemma:list-construction} ensures that consecutive elements in list $L$ are adjacent vertices in $\tree$. Then, since the parties have computed the list representation $L$ identically, the parties obtain vertices $L_{\closestInt{j}}$ that are $1$-close in $\tree$.

If the honest parties have obtained the same vertex $L_{\closestInt{j}}$, then all honest parties obtain $\pathh^{\star} = \pathh(\roott, L_{\closestInt{j}})$, and hence Property 2 holds. Otherwise, honest parties have obtained adjacent vertices $v, v'$. Assume, without loss of generality, that $v$ is the parent of $v'$ in the tree $\tree$ rooted at $\roott$.
We may then define $\pathh^\star$ as the ``longer'' path obtained by the honest parties, i.e., $\pathh^\star := \pathh(\roott, v')$: parties that have obtained vertex $v'$ return path $\pathh^\star$, while parties that have obtained vertex $v$ return path $\pathh(\roott, v)$ such that $v' \notin \vertices(\pathh(\roott, v))$ (since $v$ is the parent of $v'$) and $\pathh^\star = \pathh(\roott, v) \oplus (v, v')$. Consequently, Property 2 holds in this case as well.

Finally, the number of rounds follows from the round complexity of $\realAA$, using the fact that we run $\realAA$ with $\varepsilon := 1$ on inputs between $1$ and $\abs{L}$, and $\abs{L} \leq 2 \cdot \abs{\vertices(\tree)}$ according to \Cref{lemma:list-construction}. 
\end{proof}

\subsubsection{Protocol} \label{section:final-protocol}
We now combine the components described in \Cref{section:warmup} and \Cref{section:path} to present an almost-optimal protocol achieving $\approximateagreement$ on trees, denoted by~$\treeAAOld$.

The parties begin by fixing the root vertex ${\roott}$ of the input space tree $\tree$ as the vertex with the lowest label in lexicographic order. Afterwards, they run the subprotocol $\pathfinder$ described in Section \ref{section:path} to \emph{approximately} agree on paths that intersect the honest inputs' convex hull, as described by Lemma \ref{lemma:path-finder}. 

Once these paths are obtained, each party $\party$  proceeds according to the approach described in Section \ref{section:path-assumed}: it denotes the $k$ vertices on its own path $\pathh$ as $(v_1 := \roott, v_2, \ldots, v_k := v)$. It joins $\realAA(1)$ with input $i$, where the vertex denoted by $v_i$ in  $\party$'s path $\pathh$ is the projection of $\party$'s input vertex $v_{\inputt}$ onto $\pathh$; formally $v_i = \projection_{\pathh}(v_{\inputt})$.

Upon obtaining an output $j$ from $\realAA(1)$, party~$\party$ \emph{should} output the vertex denoted by $v_{\closestInt{j}}$ in its path $\pathh$.  However, this is the point where we need to be careful about honest parties holding different paths: an honest party $\party$ might obtain $j > k$ from $\realAA(1)$. If this is the case, then $\party$ holds the ``shorter'' path $(v_1, v_2, \ldots, v_{k})$, while other honest parties hold the ``longer'' path $\pathh^\star = (v_1, v_2, \ldots, v_{k+1})$, as described in Lemma~\ref{lemma:path-finder}. 
In this case, if $\closestInt{j} = k + 1$ and vertex $v_k$ has at least three neighbors (as illustrated in Figure \ref{figure:two-paths}) party~$\party$ cannot uniquely determine which  of the neighbors corresponds to $v_{k + 1}$ -- the last vertex of the ``longer path''.

\begin{figure}[h]
\centering
\includegraphics[width=0.7\textwidth]{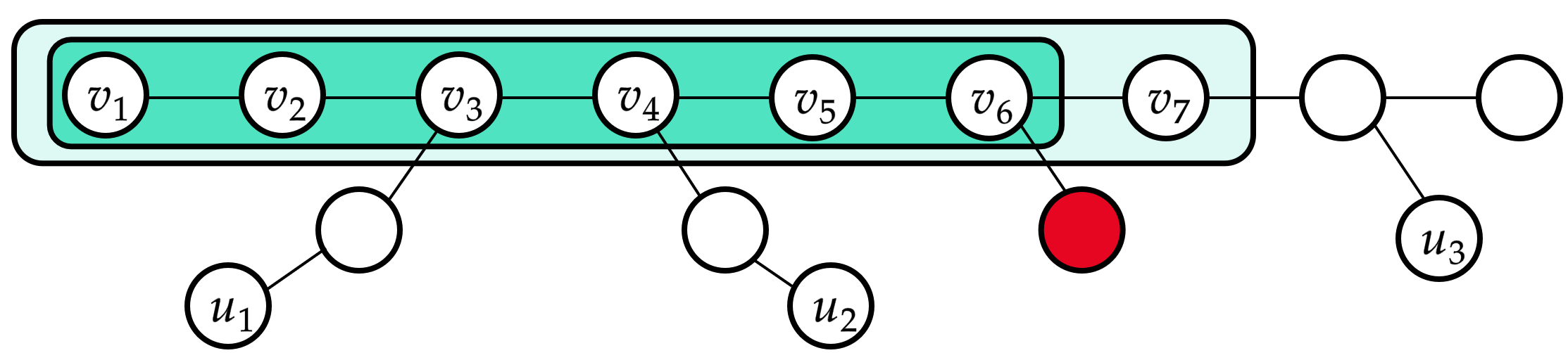}
\caption{In this figure, vertices $u_1, u_2, u_3$ are the honest inputs, and the highlighted paths 
$(v_1, v_2, \ldots, v_6)$ and $(v_1, v_2, \ldots, v_7)$ represent the paths $\pathh$ that the honest parties obtained via $\pathfinder$.
Note that an honest party $\party$ that holds $\pathh = (v_1, \ldots, v_6)$ might obtain a value $j$ such that $\closestInt{j} = 7$, and $\party$ does not know whether $v_7$ should be the actual vertex $v_7$ or the red vertex adjacent to $v_6$. The red vertex is, in fact, outside the honest inputs' convex hull.
}\label{figure:two-paths}
\end{figure}

In $\treeAAOld$, if party $\party$ obtains $j > k$, it simply outputs $v_k$. We will show that, in this case, all honest parties output either $v_{k}$ or, if they hold the longer path $P$, $v_{k + 1}$, and therefore $\approximateagreement$ is achieved.

We present the code of $\treeAAOld$ below. 

\begin{dianabox}{$\treeAAOld$}
	\algoHeadNoBold{Code for party $\party$ with input $v_{\inputt} \in \vertices(\tree)$}
	\begin{algorithmic}[1]
            \State $\roott :=$ the vertex in $\tree$ with the lowest label (in lexicographic order).
            \State  $\pathh := \pathfinder(\tree, \roott, v_{\inputt}).$ 
            \State Denote the $k = \abs{V(\pathh)}$ vertices in $\pathh$ by $(v_1 := \roott, v_2, \ldots, v_k)$.
            \State Wait until round $R_{\pathfinder}$ ends.
            \State Join $\realAA(1)$ with input $i$ such that $\projection_{\pathh}(v_{\inputt}) = v_i$; obtain output $j$.
            \State If $\closestInt{j} > k$, output the vertex denoted by $v_{k}$ in $\vertices(\pathh)$; otherwise, output the vertex denoted by $v_{\closestInt{j}}$ in $\vertices(\pathh)$.
    \end{algorithmic}
\end{dianabox}

Note that in line 4, the parties wait for a specific number of rounds, namely $R_{\pathfinder}$, before initiating the execution of $\realAA$ in line 5. This is because $\pathfinder$ -- which uses $\realAA$ as a building block -- does not guarantee that all the honest parties obtain their paths in the same round (as briefly mentioned in Section \ref{section:warmup}). 
Instead, it only guarantees that all the honest parties obtain their paths within $R_{\pathfinder}$ rounds. Hence, by waiting until the end of round $R_{\pathfinder}$, all the honest parties can start executing $\realAA$ in line 5 simultaneously.

The theorem below describes the guarantees provided by protocol~$\treeAAOld$.

\begin{theorem}
    Assume a protocol $\realAA$ achieving $\approximateagreement$ on $\mathbb{R}$ with round complexity $R_{\realAA}(\diameter, \varepsilon)$ when the honest inputs are $D$‑close and the target error is $\varepsilon$.
  
    Even when up to $t < n$ of the $n$ parties involved are Byzantine, protocol $\treeAAOld$ achieves $\approximateagreement$ on any input space tree $\tree$ within
    $R_{\realAA}(2 \cdot \abs{\vertices(\tree)}, 1) + R_{\realAA}(\diameter(\tree), 1)$ rounds.    
\end{theorem}

\begin{proof}
    We first note that \Cref{lemma:path-finder} ensures that the honest parties obtain their paths $\pathh$ within $R_{\pathfinder}$ rounds, and therefore the honest parties  start the execution of $\realAA(1)$ simultaneously in line 5 of the protocol. Then $\realAA(1)$ guarantees that the honest parties obtain $1$-close real values $j$ that are within the range of honest inputs $i$. 
    Using \Cref{remark:combined}, we may establish that the same properties hold for the integer values $\closestInt{j}$ obtained by the honest parties; these are also in the range of honest inputs $i$, and are also $1$-close.
    
    As the honest values $i$ are $\diameter(\tree)$-close, the honest parties obtain values $j$ within $R_{\realAA}(\diameter(\tree), 1)$ rounds. This implies that $\treeAAOld$ achieves Termination within $R_{\pathfinder} + R_{\realAA}(\diameter(\tree), 1)$ rounds. Afterwards, using \Cref{lemma:path-finder}, we obtain that $\treeAAOld$ achieves Termination within $ R_{\realAA}(\abs{\vertices(\tree)}, 1) + R_{\realAA}(\diameter(\tree), 1)$ rounds.

    In order to discuss Validity and $1$-Agreement, it is important to establish that, in line~3, the honest parties denote the vertices of their paths $\pathh$ in a consistent manner. According to \Cref{lemma:path-finder}, and, as depicted in \Cref{figure:two-paths}, there are two $1$-close vertices $v$ and $v'$ with the following properties:
    \begin{enumerate}[label=(\roman*)]
        \item $v$ is the parent of $v'$ (with respect to the root vertex $\roott$);
        \item every honest party obtains a path $\pathh$ such that either $\pathh = \pathh(\roott, v)$ or, if $v \neq v'$, $\pathh = \pathh(\roott, v') = \pathh(\roott, v) \oplus (v, v')$.
    \end{enumerate}
    
    Then all honest parties denote the vertices in $\pathh(\roott, v)$ identically. In addition, honest parties holding path $\pathh(\roott, v')$ denote vertex $v'$ identically.
    
    We may now discuss Validity. \Cref{lemma:path-finder} ensures that the paths $\pathh$ the honest parties obtain (hence, both $\pathh(\roott, v)$ and $\pathh(\roott, v')$) intersect the honest inputs' convex hull. Applying \Cref{remark:convex-hull-maintained}, this implies that, for every honest party, $\projection_{\pathh}(v_{\inputt})$ is a valid vertex. Hence, for every honest party's index $i$, the vertex denoted by $v_i$ in $\pathh(\roott, v')$ is a valid vertex. Moreover, if $i_{\min}$ and $i_{\max}$ represent the lowest and the highest honest indices $i$, then all vertices denoted by $v_{i_{\min}}, \ldots, v_{i_{\max}}$ in $\pathh(\roott, v')$ are valid.
    We recall that honest parties obtain integer values $\closestInt{j} \in [i_{\min}, i_{\max}]$.
    If the condition $\closestInt{j} \leq k$ holds for an honest party $\party$ (i.e., it knows which vertex is denoted by $v_{\closestInt{j}}$), then it follows immediately that $\party$ outputs a valid vertex $v_{\closestInt{j}}$. Otherwise, if an honest party $\party$ has obtained $\closestInt{j} > k$, then $\party$ holds the \emph{shorter} path $\pathh$: $v \neq v'$, $\party$ holds the path $\pathh(\roott, v)$ of $k^\star$ vertices. Moreover, there is an honest party $\party'$ that has obtained the path $\pathh(\roott, v) \oplus (v, v')$ of $k^\star + 1$ vertices, and has joined $\realAA(1)$ with index $i = k^\star + 1$. This implies that $k^\star \in [i_{\min}, i_{\max}]$ and therefore $\party$'s vertex $v_k$ (where $k = k^\star$) is a valid vertex. Consequently, $\treeAAOld$ achieves Validity.

    For $1$-Agreement, we recall that honest parties obtain $1$-close integer values $\closestInt{j}$. Therefore, the vertices denoted by $v_{\closestInt{j}}$ in $\pathh(\roott, v')$ are $1$-close. This implies that, if the condition $\closestInt{j} \leq k$ is satisfied for all honest parties, $1$-Agreement is achieved. We still need to discuss the case where some honest party $\party$ obtains $\closestInt{j} > k$: again, this means that  $v \neq v'$, $\party$ holds the path $\pathh(\roott, v)$ of $k^\star$ vertices, and there is an honest party $\party'$ that holds the path $\pathh(\roott, v) \oplus (v, v')$ of $k^\star + 1$ vertices. In this case, we show that all honest parties obtain $\closestInt{j} \in \{k^\star, k^\star + 1\}$.
    First, since the honest parties' integer values $\closestInt{j}$ are $1$-close, every honest party has obtained $\closestInt{j} \geq k^\star$.
    Second, \Cref{lemma:path-finder} ensures that no honest party holds a path $\pathh$ consisting of more than $k^\star + 1$ vertices, which implies that $i_{\max} \leq k^\star+1$. As the honest parties' integer values $\closestInt{j}$ are in $[i_{\min}, i_{\max}]$, every honest party obtains $\closestInt{j} \leq  k^\star + 1$. Therefore, every honest party obtains $\closestInt{j} \in \{k^\star, k^\star + 1\}$.  
    Hence, as every honest party holds either $k = k^\star$ or $k = k^\star + 1$, every honest party outputs either $v_{k^\star}$ or $v_{k^\star + 1}$, so $1$-Agreement holds. We may therefore conclude that $\treeAAOld$ achieves $\approximateagreement$.
\end{proof}

By instantiating $\realAA$ with the protocol described in \Cref{theorem:real-values-aa}, we obtain the following:
\begin{corollary}
    There is a protocol $\treeAAOld$ achieving $\approximateagreement$ on any input space tree $\tree$ whenever up to $t < n / 3$ of the $n$ parties involved are Byzantine, within
    $O \left(\frac{\log \abs{\vertices(\tree)}}{\log \log \abs{\vertices(\tree)}}\right)$ rounds.    
\end{corollary}

Considering our lower bound presented in \cref{section:lower-bound}, this solution is not yet optimal: given that $t \in \Theta(n)$, it matches the lower bound whenever $\log \diameter(\tree) \in \Theta(\log \vertices(\tree))$, hence for trees $\tree$ with diameter  $\diameter(\tree) \in |\vertices(\tree)|^{\Theta(1)}$. Combining this with the protocol of \cite{DISC:NoRy19} that solves $\approximateagreement$ on trees within $O(\log \diameter(\tree))$ rounds, we have optimal protocols whenever the input space tree's diameter is either constant or polynomial in the number of vertices (for some constant exponent $\leq 1$). Thus, this leaves a gap for trees $\tree$ of diameter $\diameter(\tree) \in \left[\omega(1), \abs{\vertices(\tree)}^{o(1)} \right]$. In the subsequent section, we give an improved protocol that closes this gap.

\subsection{Weakening the Path Assumption: Optimal Solution} \label{section:fox-paths}

We observe that the idea of \Cref{section:path-assumed} does not require the honest parties to \emph{know the same path} that intersects the honest inputs' convex hull. 
As pointed out in \Cref{section:brief-announcement-solution}, what the argument needs is only a weak form of pre-agreement. We now explore an even weaker form of pre-agreement, which leads to our final solution for $\approximateagreement$ on trees.
Concretely, each honest party should hold two paths $P$ and $Q$, which may be distinct over the parties. However, the following should hold: (i) every honest party's path $P$ intersects the honest inputs' convex hull; (ii) every honest party's path $Q$ has every honest party's path $P$ as a prefix, as defined below, and (iii) every honest party can map the value $j$ obtained via $\realAA$ to a valid vertex $v_{\closestInt{j}}$. \Cref{figure:new-approach} depicts these assumptions.
\begin{definition}[Path Prefix]\label{def:path-prefix}
A path $P=(v_0,\dots,v_k)$ is a \emph{prefix} of a path $Q=(u_0,\dots,u_\ell)$ if $k \leq \ell$ and
$v_i = u_i$ for all $i \in \{0,\dots,k\}$.
\end{definition}

\begin{figure}[h]
\centering
\vspace{-2.5em}
\includegraphics[width=0.75\textwidth]{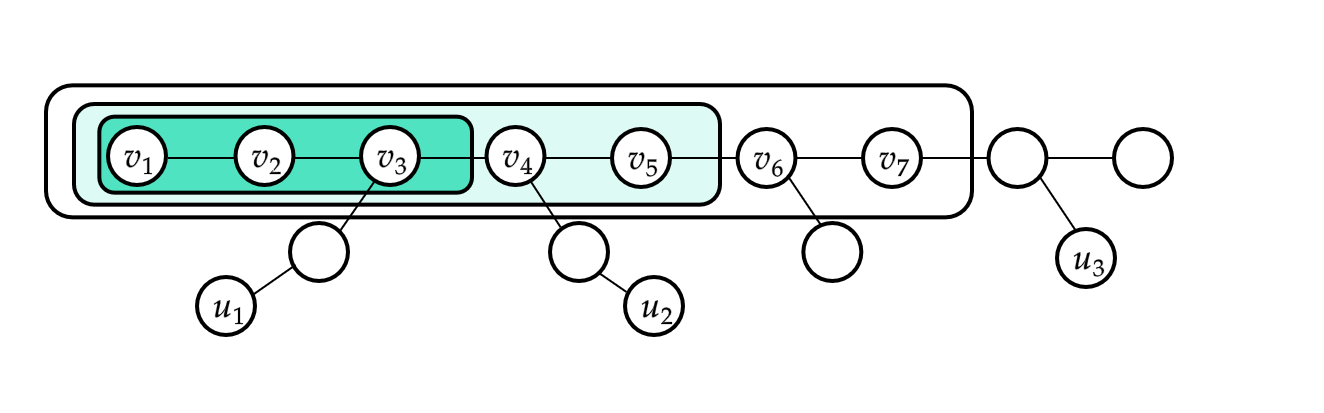}
\vspace{-2.5em}
\caption{An honest party $\party$ may hold a path $Q$ such as $(v_1, \ldots, v_7)$. Moreover, all honest parties' paths $P$ are guaranteed to be prefixes of $\party$'s path $Q$, such as those highlighted in cyan and light cyan:  $(v_1, v_2, v_3)$, and $(v_1, v_2, \ldots, v_5)$.
Note that the honest parties' paths $P$ intersect the honest inputs' convex hull $\hull{\{u_1, u_2, u_3\}}$.} \label{figure:new-approach}
\end{figure}

We explain how we can provide each honest party with such paths $P$ and $Q$ in \Cref{section:fox-path-finder}.
Given these paths, we proceed as follows: each honest party denotes the $\ell$ vertices of its path $Q$ by $(v_1,\dots,v_\ell)$. It joins $\realAA(1)$ with input $i$, where $v_i$ denotes the projection of its input vertex onto the path $P$, i.e., $\projection_{\pathh}(v_{\inputt})$. Protocol $\realAA(1)$ returns real values $j$ that are $1$-close and lie within the range of the honest inputs $i$, so by
\Cref{remark:combined}
the rounded integers $\closestInt{j}$ are also $1$-close and within that same range. Finally, each party $\party$ outputs the vertex $v_{\closestInt{j}}$ on its path $Q$: note that $\closestInt{j}$ is at most the length of the longest honest path $P$, and therefore $\party$ has assigned a label for it.

\begin{dianabox}{$\treeAA(\tree)$}
	\algoHeadNoBold{Code for party $\party$ with input $v_{\inputt} \in \vertices(\tree)$, and with paths $P$, $Q$ in $\tree$ (as defined in \Cref{lemma:treeAA-graded})}
	\begin{algorithmic}[1]
            \State Denote the $\ell = \abs{\vertices(Q)}$ vertices in the (directed) path $Q$ by $(v_1, v_2, \ldots, v_\ell)$.
            \State Join $\realAA(1)$ with input $i$ such that $\projection_{\pathh}(v_{\inputt}) = v_i$; obtain output $j$.
            \State Output the vertex denoted by $v_{\closestInt{j}}$ in $\vertices(Q)$.
    \end{algorithmic}
\end{dianabox}

We will prove that $\treeAA(\tree)$ achieves $\approximateagreement$ given that all honest parties hold such paths $P$ and $Q$. In \Cref{section:fox-path-finder} we explain how to obtain such paths, and we put everything together in \Cref{section:fox-final}.

\begin{lemma}\label{lemma:treeAA-graded}
Let $\tree$ be a labeled tree of diameter $D$.  Suppose each party is provided with two paths $P$ and $Q$ (that may be distinct across the parties) in $\tree$ such that:
\begin{enumerate}
\item every honest party's path $P$ intersects the honest inputs' convex hull;
\item every honest party's path $P$ is a prefix of its path $Q$, and if two honest parties hold $(P,Q)$ and $(P',Q')$ respectively, then $P$ is a prefix of $Q'$.
\end{enumerate}
Then, for any protocol $\realAA$ that achieves $\approximateagreement$ on $\mathbb{R}$ in $R_{\realAA}(D, \varepsilon)$ rounds when the honest inputs are $D$-close and the target error is $\varepsilon$, there is a protocol that solves $\approximateagreement$ on $\tree$ in $R_{\realAA}(D,1)$ rounds.

\end{lemma}

We split the proof of \Cref{lemma:treeAA-graded} into multiple steps. Each of the lemmas below makes the same assumptions on the honest parties' paths $P$ and $Q$ as \Cref{lemma:treeAA-graded}.
We first note that these assumptions immediately imply that the parties assign consistent labels to the vertices in their paths $Q$, as stated below. 
\begin{remark} \label{lemma:consistent-labels}
    Let $P^{\star}$ denote the longest path $\pathh$ that an honest party joins $\treeAA(\tree)$ with, and denote its vertices by $(v_1^{\star}, \ldots, v^{\star}_{k_{\max}})$. Then, for every honest party $\party$ that has labeled the vertices in its path $Q$ by $(v_1, \ldots, v_{\ell})$, the following hold: $\ell \geq k_{\max}$ and $v_i = v_i^\star$ for every $i$ with $1 \leq i \leq k_{\max}$.
\end{remark}

Next, we need to establish that the honest parties join $\realAA$ with well-defined values $i$.
\begin{lemma} \label{lemma:well-defined-overthinking}
    Let $k_{\max}$ denote the number of vertices in the longest path $\pathh$ that an honest party joins $\treeAA(\tree)$ with. Then, every honest party obtains an integer value $i$ satisfying $1 \leq i \leq k_{\max}$.
\end{lemma}
\begin{proof}
    The assumptions of \Cref{lemma:treeAA-graded} ensure that honest parties' paths $\pathh$ are non-empty, hence $\projection_{\pathh}(v_{\inputt})$ is well-defined. Then, since $\pathh$ contains $k \leq k_{\max}$ vertices denoted by $(v_1, v_2, \ldots, v_k)$, every honest party obtains a vertex $v_i = \projection_{\pathh}(v_{\inputt})$ such that $i$ is an integer satisfying $1 \leq i \leq k_{\max}$.
\end{proof}

In order to show that $\treeAA(\tree)$ achieves $1$-Agreement and Validity, we discuss the properties for the integer values $\closestInt{j}$ obtained by the honest parties in \Cref{lemma:closest-int}. Afterwards, \Cref{lemma:treeAA-outputs} translates these guarantees back to the original input space tree $T$.
\begin{lemma} \label{lemma:closest-int}
    In line 2, the honest parties obtain values $j$ such that the resulting integer values $\closestInt{j}$ are $1$-close and within the range of their values $i$.
\end{lemma}
\begin{proof}
By \Cref{lemma:well-defined-overthinking}, the honest parties obtain well-defined real values $i$, and then they simultaneously invoke $\realAA(1)$. Now, \Cref{theorem:real-values-aa} applied with $\varepsilon=1$ guarantees that the honest parties obtain $1$-close real values $j$ that lie within the range of the honest values $i$.  Finally, \Cref{remark:combined}
implies that (1) the rounded values $\closestInt{j}$ are still pairwise within distance~1, and (2)
each $\closestInt{j}$ lies in the honest range of values $i$, as claimed.
\end{proof}

\begin{lemma}\label{lemma:treeAA-outputs}
    The honest parties' outputs 
    are $1$-close vertices in the convex hull of the honest inputs.
\end{lemma}

\begin{proof}
    To prove the statement, consider an arbitrary honest party $p$ and map its output $v_{\closestInt{j}_p}$ onto the path $P^\star$ described in \Cref{lemma:consistent-labels}. Recall that $P^\star$ is the longest path $\pathh$ that some honest party has joined $\treeAA(\tree)$ with. We denote the vertices of $P^\star$ by $(v_1^{\star}, \ldots, v^{\star}_{k_{\max}})$.

    By \Cref{lemma:well-defined-overthinking}, each honest party $\party$ obtains a value $i_p$ satisfying $1\le i_p\le k_{\max}$. 
    Moreover, \Cref{lemma:closest-int} ensures that for any two honest parties $\party$ and $\party'$, the integer values $\closestInt{j}_p$ and $\closestInt{j}_{p'}$ are $1$-close. 
    Let $i_{\min} := \min_h i_h$ and $i_{\max} := \max_h i_h$ be the minimum and maximum $i_h$ value over all honest parties $h$. 
    For every honest $p$, we have that $i_{\min} \le \closestInt{j}_p \le i_{\max} \le k_{\max}$. Furthermore, \Cref{lemma:consistent-labels} implies that for every honest party $\party$ and every $k\le k_{\max}$ we have $v_k=v^\star_k$. Hence $v_{\closestInt{j}_p}=v^\star_{\closestInt{j}_p}$, and therefore, the honest parties' outputs are pairwise $1$-close along $P^\star$.

    For Validity, by the assumptions of \Cref{lemma:treeAA-graded}, each honest party’s path $P_p$ intersects convex hull of the honest inputs. By \Cref{remark:convex-hull-maintained}, every projection vertex $v_{i_p}$ lies in the convex hull of honest inputs. Consequently, for every index $k$ satisfying $i_{\min} \le k \le i_{\max}$, the vertex $v^\star_k$  is in the convex hull of the honest inputs. Since $\closestInt{j}_p$ lies within the range of honest values, we obtain that each honest party $\party$ outputs the vertex  $v_{\closestInt{j}_p}=v^{\star}_{\closestInt{j}_p}$ in the convex hull of the honest inputs.
\end{proof}

We conclude the subsection by presenting the proof of \Cref{lemma:treeAA-graded}.

\begin{proof}[Proof of \Cref{lemma:treeAA-graded}]
We first discuss Termination. Given that each party is provided with two paths $P$ and $Q$ with the properties described in the lemma's statement, \Cref{lemma:well-defined-overthinking} guarantees each honest party obtains a well-defined value $i$ with $1 \le i \le k_{\max}\le \diameter$.  The parties then invoke $\realAA(1)$ (simultaneously) on those values, and hence the total round complexity is $R_{\realAA}(\diameter, 1)$.
The $1$-Agreement and Validity properties follow directly from \Cref{lemma:treeAA-outputs}.  Therefore, \(\treeAA(\tree)\) satisfies all requirements of \(\approximateagreement\).
\end{proof}

\subsection{Finding These Paths}\label{section:fox-path-finder}
    We now present our protocol $\foxpathfinder(\tree)$, providing the parties with paths satisfying the preconditions of $\treeAA$ (described in \Cref{lemma:treeAA-graded}) within $O(1)$ rounds.
    
    Each party joins $\foxpathfinder(\tree)$ with its input vertex $v_{\inputt}$. Then, we fix some vertex $v_{\startt}$ (e.g., the vertex in $\tree$ with the lowest label in lexicographic order) and each party communicates the path from $v_{\startt}$ to its input vertex $P(v_{\startt}, v_{\inputt})$. Each party will try to define its path $P$ roughly as the \emph{longest common prefix} of the paths $P(v_{\startt}, v_{\inputt})$ it has received: note that the honest parties might not have an identical view over the paths $P(v_{\startt}, v_{\inputt})$ due to the Byzantine parties. Then, each party will define its path $Q$ roughly as what it \emph{believes} other honest parties might have identified as a longest common prefix.
    We formalize this idea with the help of a mechanism called Gradecast ($\gradecast$), which ensures that parties distribute their paths in a consistent manner, i.e., Byzantine parties are unable to send different messages to different parties, but they can still cause inconsistencies by making sure that only a subset of the honest parties receive their messages. We note that $\gradecast$ delivers each message with a grade in $\{0, 1, 2\}$. Honest parties' messages are always delivered with grade $2$. In addition, if two parties receive messages with grades $g$ and $g'$ from a party $\party$, then $\abs{g - g'} \leq 1$. Moreover, if both $g$ and $g'$ are greater than zero, then the two parties have received the same message from $\party$.
     We state the formal definition of $\gradecast$ below, as presented in \cite{BenDoHo10}.
\begin{definition}(Gradecast)\label{def:gradecast}
    Let $\Pi$ be a protocol where a designated party $S$ (the sender) holds a value $v_S$. We say that $\Pi$ achieves $\gradecast$ if the following properties hold even when $t$ out of the $n$ parties involved are Byzantine: \textbf{(Termination)} Every honest party eventually outputs a pair $(v,g)$, where $v$ is a message and $g\in\{0,1,2\}$; \textbf{(Integrity)} If $S$ is honest, then every honest party outputs $(v_S,2)$; \textbf{(Consistency)} If two honest parties output $(v,g)$ and $(v',g')$ then $\abs{g - g'} \leq 1$. In addition, if $g, g' > 0$, then $v = v'$.
\end{definition}

We then make use of the $\gradecast$ protocol of \cite{BenDoHo10} (presented in the full version \cite{BDH10}), described below.
\begin{theorem}[Theorem 1 of \cite{BDH10}]\label{theorem:gradecast}
    There is a protocol $\gradecastprotocol$ achieving $\gradecast$ even when up to $t < n/3$ of the $n$ parties involved are Byzantine, with round complexity $R_{\gradecastprotocol} := 3$. 
\end{theorem}

Hence, each party will distribute its path $P(v_{\startt}, v_{\inputt})$ via $\gradecastprotocol$. Every party obtains at least $n - t$ paths with grade $2$ (from honest parties), plus up to $t$ paths with various grades. Then, each party defines its path $P$ as a path prefix \emph{supported} by $n - t$ parties with grade $2$, as defined below. 

\begin{definition}[Supported Path Prefix]\label{def:path-prefix-support}
Let $\{P_1, P_2, \dots, P_k\}$ be the collection of paths.
We say that a path $P$ is \emph{supported by $\ell$ parties} if $\bigl|\{i : 1 \le i \le k, P\text{ is a prefix of }P_i\}\bigr|
 \ge \ell.$
The \emph{path prefix supported by $\ell$ parties} is defined to be the longest path $P$ satisfying the above property.
\end{definition}

The paths $Q$ are defined as path prefixes supported by $n - t$ parties with grade at least $1$ instead of $2$, i.e., some honest parties may have received these supports with grade $2$.

\begin{dianabox}{$\foxpathfinder(\tree)$}
	\algoHeadNoBold{Code for party $\party$ with input $v_{\inputt} \in \vertices(\tree)$}
	\begin{algorithmic}[1]
            \State Send $P(v_{\startt}, v_{\inputt})$ to all parties via $\gradecastprotocol$, where $v_{\startt}$ is the vertex in $\tree$ with the lexicographically smallest label.
            \State Let $P: =$ the longest path prefix supported by at least $n - t$ parties with grade $2$.
            \State Let $Q :=$ the longest path prefix supported by at least $n - t$ parties with grade at least $1$.
            \State Output $(P, Q)$.
    \end{algorithmic}
    \label{alg:foxpathfinder}
\end{dianabox}

The lemma below presents the guarantees of $\foxpathfinder(\tree)$.
\begin{restatable}{lemma}{foxPathFinderLemma}\label{lemma:fox-path-finder}
Assume a protocol $\gradecastprotocol$ that achieves $\gradecast$ within $R_{\gradecastprotocol}$ rounds. Then, if $t < n/2$, $\foxpathfinder(\tree)$ provides each honest party within $R_{\gradecastprotocol}$ rounds with two non-empty paths $P$ and $Q$ such that: (1) $P$ intersects the honest inputs' convex hull; (2) $P$ is a prefix of $Q$.  Moreover, if two honest parties obtain $(P,Q)$ and $(P',Q')$ respectively, then $P$ is a prefix of $Q'$.
\end{restatable}
\begin{proof}
    Let $\lcp$ be the longest common prefix of the paths $P(v_{\startt},v_{\inputt})$ held by the honest parties. \Cref{lemma:path-lcp}, stated below, ensures that $\lcp$ is non-empty and ends at a valid vertex $v_{\lcp}$.

    For the first property, we show that every honest party's path $P$ has $\lcp$ as a prefix. 
    We recall that every honest party $\party$ defines $P$ as the longest path prefix supported by at least $n - t$ paths received with grade $2$. Since the honest parties' paths have prefix $\lcp$, $\party$ receives at least $n - t$ paths with prefix $\lcp$ with grade $2$ due to the definition of $\gradecast$ (\Cref{def:gradecast}). In addition, the honest parties receive up to $t < n - t$ (since $t < n/2$) paths that do not have $\lcp$ as a prefix. Hence, the longest path prefix supported by at least $n - t$ parties with grade $2$ in $\party$'s view must have $\lcp$ as a prefix, and therefore $P$ has prefix $\lcp$. As $\lcp$ ends at a valid vertex $v_{\lcp}$, we may conclude that $P$ intersects the honest inputs' convex hull.

    For the second property, assume that an honest party $\party$ has obtained a path $P$ in line 2, and an honest party $\party'$ has obtained a path $Q'$ in line 3. We prove that $P$ is a prefix of $Q'$: we show that $\party'$ receives support from at least $n - t$ parties for the path prefix $P$, and does not receive sufficient support for any other path prefix.
    We recall once again that the path $P$ was chosen by $\party$ as the longest prefix supported by at least $n-t$ paths received with grade $2$. 
    Following the definition of $\gradecast$, we know that if $p$ receives a message with grade $2$, all honest parties get the same message with grade at least $1$.
    This guarantees that $\party'$ has received at least $n - t$ paths with grade at least $1$ that have $P$ as a prefix, and up to $t < n - t$ (since $t < n/2$) messages supporting any path prefix different from $P$. Hence, $\party'$ defines $Q'$ as a path with prefix $P$. Note that this holds if $\party = \party'$ as well.

    We still need to discuss the round complexity: $\foxpathfinder(\tree)$ makes $n$ parallel invocations of $\gradecastprotocol$, and therefore terminates within $R_{\gradecastprotocol}$ rounds.
\end{proof}

\begin{lemma}\label{lemma:path-lcp}
    Let $S \subseteq \vertices(\tree)$ 
    and consider two vertices $v_{\last}, v'_{\last}$ in $\hull{S}$. Then, for any $v_{\startt} \in \vertices(\tree)$, the longest common prefix of $P(v_{\startt}, v_{\last})$ and $P(v_{\startt}, v'_{\last})$ is non-empty and ends in $\hull{S}$.
\end{lemma}
\begin{proof}
    Let $P(v_{\startt}, v)$ be the longest common prefix of $P(v_{\startt}, v_{\last})$ and $P(v_{\startt}, v'_{\last})$. This is non-empty, as both paths start with vertex $v_{\startt}$.
    Note that, by rooting the tree $\tree$ at $v_{\startt}$, vertex $v$ becomes the lowest common ancestor of $v_{\last}$ and $v'_{\last}$ and therefore lies on the unique path $P(v_{\last}, v'_{\last})$. Then, as $v_{\last}, v_{\last}' \in \hull{S}$, we obtain that $v \in \hull{S}$. 
\end{proof}

\subsection{Putting it all together} \label{section:fox-final}

In the following, we present our final protocol for $\approximateagreement$ on trees. The parties run $\foxpathfinder$, which provides them with paths meeting the preconditions of $\treeAA$. Afterwards, the parties run $\treeAA$, which provides them with the final outputs.

\begin{dianabox}{$\finalTreeAA(\tree)$}
	\algoHeadNoBold{Code for party $\party$ with input $v_{\inputt} \in \vertices(\tree)$}
	\begin{algorithmic}[1]
            \State  Join $\foxpathfinder(\tree)$ with input $v_{\inputt}$ and obtain paths $P$ and $Q$.
            \State Run $\treeAA(\tree)$ with inputs $v_{\inputt}$, $P$, $Q$ to compute $v_{\outputt}$; then output $v_{\outputt}$ and terminate.
    \end{algorithmic}
\end{dianabox}

The theorem below follows directly from \Cref{lemma:fox-path-finder} and \Cref{lemma:treeAA-graded}.
\begin{theorem}\label{theorem:opt-aa}
    Assume a $\gradecast$ protocol $\gradecastprotocol$ with round complexity $R_{\gradecastprotocol}$, and a protocol $\realAA$ achieving $\approximateagreement$ on $\mathbb{R}$ with round complexity $R_{\realAA}(\diameter, \varepsilon)$ when the honest inputs are $D$‑close and the target error is $\varepsilon$.
    Then, if $t < n/2$, $\finalTreeAA(\tree)$ achieves $\approximateagreement$ on any input space tree $\tree$ within $R_{\gradecastprotocol} + R_{\realAA}(\diameter(\tree), 1)$ rounds.
\end{theorem}

We obtain our final corollary by instantiating $\gradecastprotocol$ with the protocol described in \Cref{theorem:gradecast} and $\realAA$ with the protocol described in \Cref{theorem:real-values-aa}.
\begin{corollary} \label{cor:tree-aa}
    There is a synchronous protocol that achieves $\approximateagreement$ on any input space tree $\tree$ within $O\bigl(\tfrac{\log \diameter(\tree)}{\log\log \diameter(\tree)}\bigr)$ rounds even when up to $t < n/3$ of the $n$ parties involved are Byzantine.
\end{corollary}

\paragraph{A note on the $t < n / 2$ case.}
We have presented our protocol $\finalTreeAA(\tree)$ under the assumption that $t < n / 3$: this is the optimal resilience threshold in an \emph{unauthenticated setting} (i.e., no cryptographic assumptions)~\cite{JACM:DLPSW86}. This setting was chosen solely for simplicity of presentation. In fact, our reduction can be adapted to honest-majority \emph{authenticated} settings as well. The term \emph{authenticated} refers to settings assuming a public key infrastructure (PKI) and a secure digital signature scheme. Concretely, we may instantiate $\gradecastprotocol$ with the protocol of \cite{MicVai17}, described in \Cref{theorem:gradecast-pki}. 
\begin{theorem}[Theorem 4.1 of \cite{MicVai17}]\label{theorem:gradecast-pki}
    There is a protocol $\gradecastprotocol^\pki$ achieving $\gradecast$ in an authenticated setting even when up to $t < n/2$ of the $n$ parties involved are Byzantine, with round complexity~$R_{\gradecastprotocol^\pki} := 3$.
\end{theorem}

For $\realAA$, we obtain the result below by combining two protocols: the $\approximateagreement$ protocol of \cite{PODC:GhLiWa24} and the Proxcensus protocol of \cite{EUROCRYPT:GhGoLi22}.  
After making minor adjustments, the protocol of \cite{EUROCRYPT:GhGoLi22} becomes an $\approximateagreement$ protocol with round complexity $O \left( \frac{\log(\diameter/\varepsilon)}{\log\log(\diameter/\varepsilon)} \right)$ whenever $\diameter / \varepsilon = \omega(1)$. If $\diameter / \varepsilon = O(1)$ instead, we make use of the protocol of \cite{PODC:GhLiWa22}, which achieves Termination within $O(\log (\diameter/\varepsilon)) = O(1)$ rounds. We present this formally in \Cref{appendix:real-values-pki}.
\begin{restatable}{theorem} {RealAAPKI}\label{theorem:realAA-pki}
    If the honest inputs are $\diameter$-close real values, there is a protocol $\realAA^\pki(\diameter, \varepsilon)$ achieving $\approximateagreement$ in authenticated settings
    even when up to $t = c \cdot n$ (where $c$ is a constant in $[1/3, 1/2)$) of the $n$ parties are Byzantine. $\realAA^\pki(\diameter, \varepsilon)$ achieves Termination within
    $
      R_{\realAA^\pki}(\diameter, \varepsilon)
    = 
    O \left( \frac{\log(\diameter/\varepsilon)}{\log\log(\diameter/\varepsilon)}
    \right)
    $
rounds.
\end{restatable}

This yields a protocol with optimal round complexity up to $t = c \cdot n$ corruptions for any constant $c \in [1/3, 1/2)$, assuming a public key infrastructure and digital signatures, i.e., in the authenticated setting. This leads to the following corollary.
\begin{corollary} \label{cor:pki-tree-aa}
    There is a synchronous protocol that solves $\approximateagreement$ on any input space tree $\tree$ in $O\bigl(\tfrac{\log \diameter(\tree)}{\log\log \diameter(\tree)}\bigr)$ rounds in authenticated settings even when up to $t = c \cdot n$ of the $n$ parties are Byzantine, where $c$ is a constant in $[1/3, 1/2)$.
\end{corollary}
We add that the cryptographic assumptions lead to a protocol secure against a \emph{computationally bounded} adversary. When the digital signatures are replaced with real-world instantiations, the security of the protocol holds except with negligible probability (in the digital signature scheme's security parameter). 

\section{Extension to Block Graphs}
In this section, we extend our algorithm from trees to connected block graphs. Our approach utilizes the clique tree representation to reduce the problem on a block graph to a problem on a tree, enabling us to achieve optimal round complexity for block graphs $\graph$.  We first give the definition of block graphs and their clique tree structure, and afterwards present the reduction, and the resulting corollaries in the synchronous and asynchronous models.

\paragraph{Block Graphs.}
A graph $\graph =(V, E)$ is a \emph{block graph} if every maximal $2$-connected subgraph (block) of $\graph$ is a clique. Equivalently, $\graph$ is a chordal graph in which every two maximal cliques intersect in at most one vertex. 

\begin{figure}[ht]
\centering
\begin{tikzpicture}[
  v/.style={circle, draw=black, fill=white,  thick, inner sep=2.8pt},
  r/.style={circle, draw=black, fill=cyan!20!lime, thick, inner sep=2.8pt},
  thickedge/.style={draw=black, thick},
  block/.style={rectangle, draw=teal!80, fill=teal!10, thick, minimum size=6mm, rounded corners},
  tree_edge/.style={draw=teal!50!black, thick},scale=.7
]

\begin{scope}[local bounding box=leftgraph]

    \node[v] (a1) at (-6,2) {};
    \node[v] (a2) at (-5,3) {};
    \node[v] (a3) at (-4,2) {};
    \node[r] (a4) at (-5,1) {}; 
    \foreach \i/\j in {a1/a2,a2/a3,a3/a4,a4/a1,a1/a3,a2/a4} \draw[thickedge] (\i)--(\j);

    \node[v] (b1) at (-6,-1) {};
    \node[v] (b2) at (-5,-2) {};
    \node[r] (b3) at (-4,-1.2) {}; 
    \node[v] (b4) at (-4.5,0) {};
    \node[r] (b5) at (-5.5,0) {};  
    \foreach \i/\j in {b1/b2,b1/b3,b1/b4,b1/b5,b2/b3,b2/b4,b2/b5,b3/b4,b3/b5,b4/b5} \draw[thickedge] (\i)--(\j);

    \draw[thickedge] (a4)--(b5);

    \node[v] (g1) at (-6.7,0.8) {};
    \node[v] (g2) at (-6.9,-.5) {};
    \foreach \i in {g1,g2} \draw[thickedge](\i)--(b5);
    
    
    \node[v] (c1) at (-3.2,0) {};
    \node[r] (c2) at (-2.5,-1.2) {}; 
    \draw[thickedge] (b3)--(c2);
    \draw[thickedge] (b3)--(c1)--(c2);

    \node[v] (f1) at (-3.1,-2.2) {}; 
    \node[v] (f2) at (-1.9,-2.2) {}; 
    \node[v] (f3) at (-2.5,-3.2) {}; 
    \foreach \i/\j in {f1/f2,f1/f3,f2/f3} \draw[thickedge] (\i)--(\j);
    \foreach \i in {f1,f2,f3} \draw[thickedge] (\i)--(c2);
    
    \node[r] (r5) at (-1,-0.6) {}; 
    \node[v] (d1) at (-.2,-1.6) {};
    \node[v] (d2) at (.8,-0.6) {};
    \node[r] (d3) at (-.2,0.6) {}; 
    \foreach \i/\j in {d1/d2,d2/d3,d3/d1} \draw[thickedge] (\i)--(\j);
    \foreach \i in {d1,d2,d3} \draw[thickedge] (r5)--(\i);

    \draw[thickedge] (c2)--(r5);

    \node[v] (e1) at (-1.4,1.8) {};
    \node[r] (e2) at (-.2,2.3) {};
    \node[v] (e3) at (1,1.8) {};
    \node[v] (e4) at (0.6,.9) {};
    \node[v] (e5) at (-1,.9) {};
    \foreach \i/\j in {e1/e2,e1/e3,e1/e4,e1/e5,e2/e3,e2/e4,e2/e5,e3/e4,e3/e5,e4/e5} \draw[thickedge] (\i)--(\j);
    \foreach \i in {e1,e2,e3,e4,e5} \draw[thickedge] (d3)--(\i);
    
\end{scope}

\draw[dashed, gray] (1.5, 3.5) -- (1.5, -3);

\begin{scope}[xshift=8cm, yshift=0cm]

    \node[block] (k4) at (-4, 2.5) {$K_4$};
    \node[block] (k22) at (-4, .5) {$K_2$};
    \node[block] (k222) at (-5.8, .9) {$K_2$};
    \node[block] (k2222) at (-5.8, -.5) {$K_2$};
    \node[block] (k5) at (-5, -1.5) {$K_5$};
    \node[block] (k3)  at (-3, -1.5) {$K_3$};
    \node[block] (k444)  at (-3, -3.2) {$K_4$};
    \node[block] (k2)  at (-1.5, -1.5) {$K_2$};
    \node[block] (k44) at (0, 0) {$K_4$};
    \node[block] (k6)at (0, 2) {$K_6$};

    \draw[tree_edge](k22)--(k222);
    \draw[tree_edge](k22)--(k2222);

    \draw[tree_edge] (k4) -- (k22);
    \draw[tree_edge] (k22) -- (k5);
    \draw[tree_edge] (k5) -- (k3);
    \draw[tree_edge] (k3) -- (k2);
    \draw[tree_edge] (k3) -- (k444);
    \draw[tree_edge] (k2) -- (k44);
    \draw[tree_edge] (k44) -- (k6);
    \draw[tree_edge] (-1.121,-1.121)--(-.379,-.379);
\end{scope}

\end{tikzpicture}
\caption{A block graph and its clique tree representation.}
\label{fig:block-graph-tree}
\end{figure}

\paragraph{Clique trees.} For any chordal graph $\graph$, hence also for any block graph $\graph$, its maximal cliques can be arranged as the vertices of a tree $\tree$, called the \emph{clique tree} \cite{IPL:BerPog11}, as shown in \Cref{fig:block-graph-tree}. We use $\vertices_{\graph}(c)$ to denote the set of vertices in the clique vertex $c \in \vertices(\tree)$. If clique vertices $c, c'$ are adjacent in $\tree$, then $\vertices_{\graph}(c) \cap \vertices_{\graph}(c') \neq \emptyset$. The edges of $\tree$ represent a subset of the non-empty intersection relation between maximal cliques in $\graph$, selected such that the \emph{Running Intersection Property} holds: for any two maximal cliques $c, c'$, the intersection $\vertices_\graph(c) \cap \vertices_\graph(c')$ is contained in every clique vertex on the unique path connecting them in $\tree$. Such a clique tree of minimum diameter can be computed deterministically and efficiently in $O(|V| + |E|)$ time for any chordal graph~\cite{CliqueTreesAlg}. 

The following lemma states an outcome of the running intersection property on block graphs: every vertex in the intersection of at least two maximal cliques of $\graph$ is a cut vertex in graph $\graph$.

\begin{restatable}[Clique-Tree Edge Separator]{lemma}{CUTVERTEX}\label{lemma:block-graphs:clique-separator} Let $G$ be a connected block graph and let $T$ be a clique tree of $G$. Let $\{c,c'\}$ be an edge of $T$ and let $ u \in \vertices_{\graph}(c)\cap \vertices_{\graph}(c').$ Removing the edge $\{c,c'\}$ partitions $T$ into two connected components $T_{c}$ and $T_{c'}$ containing $c$ and $c'$, respectively. Let $V_{c} := \bigcup_{v \in\vertices(T_{c})} \vertices_{\graph}(v)$ and $V_{c'} := \bigcup_{v \in\vertices(T_{c'})} \vertices_{\graph}(v).$ Then for any $x \in V_{c} \setminus \{u\}$ and any $y \in V_{c'} \setminus \{u\}$, every path in $G$ from $x$ to $y$ passes through $u$. Consequently, $u$ is a cut vertex of $G$. 
\end{restatable}
\begin{proof}
The proof of this lemma is given in the appendix (see~\Cref{sec:appndix_proof_of_lemma}).
\end{proof}

\paragraph{Subdivided Clique Tree.} 
Our protocol relies on a \emph{subdivided clique tree} $\tree_S$ of $\graph$. Given a clique tree $\tree$ of $\graph$, the subdivided clique tree $\tree_S$ is obtained by subdividing every edge $\{c,c'\}$ of $\tree$ as follows: the edge $\{c,c'\}$ is replaced by two edges $\{c,c''\}$ and $\{c'',c'\}$, where $c''$ is a newly introduced intermediate vertex corresponding to the unique vertex in which the cliques $c$ and $c'$ intersect.
By construction, the diameter of the subdivided clique tree satisfies $\diameter(\tree_S) = 2\cdot\diameter(\tree)$. Note that this is similar to the \emph{extended clique tree} defined by \cite{DISC:NoRy19}.

\paragraph{Protocol.}
We now present our $\approximateagreement$ protocol for connected block graphs. First, note that the labeled graph $\graph$ is known to all parties. The parties then compute the same minimum-diameter labeled clique tree $\tree$ for $\graph$ using the algorithm guaranteed by 
~\cite{CliqueTreesAlg}, and hence the same minimum-diameter labeled subdivided clique tree $\tree_S$. The parties then run $\finalTreeAA(\tree_S)$ using as inputs the clique vertices in $\tree_S$ containing their input vertices from $\graph$. If a vertex belongs to multiple cliques of $\graph$, the party selects arbitrarily one of the corresponding clique vertices in $\tree_S$. This provides each party with a vertex $c_{\outputt}$ in $\tree_S$ which is either an intermediate vertex or a clique vertex. Each party then maps $c_{\outputt}$ to a vertex in the original graph $\graph$, which will be its output. We present the code below, and \Cref{theorem:opt-block-graphs} presents the guarantees.
\begin{dianabox}{$\blockgraphAA(\graph)$}
	\algoHeadNoBold{Code for party $\party$ with input $v_{\inputt} \in \vertices(\graph)$}
	\begin{algorithmic}[1]
            \State Let $\tree_S$ be a minimum-diameter labeled subdivided clique tree of $\graph$.
            \State Let $c_\inputt$ be an arbitrarily chosen vertex in $\tree_S$ such that $v_{\inputt} \in \vertices_{\graph}(c_\inputt)$.
            \State Join $\finalTreeAA(\tree_S)$ with input vertex $c_\inputt$. Obtain output vertex $c_\outputt \in \vertices(\tree_S)$.
            \State
            If $c_{\outputt}$ is a clique vertex, output $v_{\outputt} := $ the vertex $v \in \vertices_\graph(c_{\outputt})$ with the shortest distance to $v_{\inputt}$.
            \State Otherwise, $c_{\outputt}$ is an intermediate vertex. Let $c, c' \in \vertices(\tree_S)$ be the two vertices adjacent to $c_{\outputt}$ in $\tree_S$, and let $v_{\outputt}$ be the unique vertex in $\vertices_\graph(c) \cap \vertices_\graph(c')$. Output $v_{\outputt}$.            
    \end{algorithmic}
\end{dianabox}

\begin{theorem}\label{theorem:opt-block-graphs}
    Assume a protocol $\finalTreeAA(\tree)$ achieving $\approximateagreement$ on any tree $\tree$ within $R_{\finalTreeAA}(\diameter(\tree))$ rounds. Then, $\blockgraphAA(\graph)$ achieves $\approximateagreement$ on any connected block graph $\graph$ within $R_{\finalTreeAA}(4 \cdot \diameter(\graph))$ rounds.
\end{theorem}

We split the proof of \Cref{theorem:opt-block-graphs} into multiple lemmas. The lemma below establishes Validity.
\begin{lemma}\label{lemma:block-graphs:validity}
    Every honest party $\party$ obtains a valid output vertex.
\end{lemma} 
\begin{proof}
Let $c_{\inputt}\in\vertices(\tree_S)$ be the clique vertex corresponding to $\party$'s input (so $v_\inputt \in V_\graph(c_\inputt)$). By the Validity of $\finalTreeAA(\tree_S)$,  the output $c_{\outputt}$ obtained by $\party$ lies in the convex hull of the honest inputs in $\tree_S$. That is, $c_\outputt$ lies on a path in $\tree_S$ between two honest clique vertices $c_a$ and $c_b$ -- these represent two maximal cliques in $\graph$.
We split the proof into three cases, reflecting the output rule in the protocol. In the following, $\tree$ denotes the clique tree of $\graph$ from which $\tree_S$ is constructed.

\paragraph{Case 1: $c_{\outputt}$ is an intermediate vertex, i.e., $c_{\outputt} \in \vertices(\tree_S) \setminus \vertices(\tree)$.}
In this case, $c_{\outputt}$ subdivides an edge $\{c,c'\}$ of $\tree$. Party $\party$ outputs the unique vertex $v_{\outputt}$ in $\vertices_\graph(c) \cap \vertices_\graph(c')$. Since $c_\outputt$ lies on the unique path between $c_a$ and $c_b$ in $\tree_S$, the edge $e = \{c,c'\}$ must lie on the unique path between $c_a$ and $c_b$ in $\tree$. Removing the edge
$e$ separates $\tree$ into two connected components, $\tree_a$ (containing $c_a$) and $\tree_b$ (containing $c_b$). By \Cref{lemma:block-graphs:clique-separator}, this will imply that any path in $\graph$ from an input vertex in $c_a$ to an input vertex in $c_b$ must pass through $v_\outputt$. Thus, $v_\outputt$ lies on a shortest path between two honest inputs in $\graph$, and therefore lies in the convex hull of the honest inputs.

\paragraph{Case 2: $c_{\outputt}$ is a clique vertex, and $v_{\inputt} \in \vertices_\graph(c_{\outputt})$.} Here, $\party$ outputs $v_{\outputt} = v_{\inputt}$, which is trivially valid. 

\paragraph{Case 3: $c_{\outputt}$ is a clique vertex, but $v_{\inputt} \notin \vertices_\graph(c_{\outputt})$.}
Let $c := c_{\outputt}$. By the Validity guarantees of $\finalTreeAA(\tree_S)$, there are two honest inputs $c_a, c_b$ in $\tree_S$ so that $c$ lies on the unique path in $\tree_S$, and hence also in $\tree$, connecting $c_a$ and $c_b$. Then, let $c'$ be the neighbor of $c$ lying on the path from $c$ to $c_b$. Removing the edge $e = \{c, c'\}$ from $\tree$ partitions $\tree$ into two components: $\tree_a$ (containing $c_a$ and $c$) and $\tree_b$ (containing $c'$ and $c_b$). 
By \Cref{lemma:block-graphs:clique-separator} any path in $\graph$ from the honest inputs $v_a$ and $v_b$ contained by the cliques $c_a$ and $c_b$ must pass through a vertex $u \in \vertices_{\graph}(c) \cap \vertices_{\graph}(c').$ As $\vertices_{\graph}(c) \cap \vertices_{\graph}(c')$ only contains one vertex, $v_{\outputt} = u$. Consequently, $v_{\outputt}$ lies on the shortest path between two honest inputs, and therefore is valid.
\end{proof}

We now show that the $1$-Agreement property holds. Throughout the following, we use subscripts in the distance notation to distinguish between distances measured in the original graph $\graph$ and distances measured in the subdivided clique tree $\tree_S$.

\begin{lemma} \label{lemma:block-graphs:1-agreement}
    If two honest parties $\party$ and $\party'$ output $v_{\outputt}$ and $v'_{\outputt}$ respectively, $\distance_\graph(v_{\outputt}, v'_{\outputt}) \leq 1$.
\end{lemma}
\begin{proof}

Let $c_{\outputt}$ and $c'_{\outputt}$ be the outputs that parties $\party$ and $\party'$ obtain from running $\finalTreeAA(\tree_S)$.
By the $1$-Agreement property of $\finalTreeAA(\tree_S)$, we have $\distance_{\tree_S}(c_{\outputt}, c'_{\outputt}) \le 1.$
We split the analysis into two cases.

\paragraph{Case 1: $\distance_{\tree_S}(c_{\outputt}, c'_{\outputt}) = 0$.}
If $c_{\outputt}$ is an intermediate vertex, i.e., a subdivision vertex of some edge $\{c,c'\}$ of $\tree$, then both parties output the same vertex. Specifically, $\{v_{\outputt}\} = \{v'_{\outputt}\} = \vertices_\graph(c) \cap \vertices_\graph(c')$, implying $\distance_{\graph}(v_{\outputt}, v'_{\outputt}) = 0$.
Otherwise, $c_{\outputt}$ is a clique vertex. Then, both parties output vertices $v_{\outputt}, v'_{\outputt} \in \vertices_\graph(c_{\outputt})$. Since $\vertices_\graph(c_{\outputt})$ induces a clique in $\graph$, any two vertices in it are either identical or adjacent; therefore, $\distance_\graph(v_{\outputt}, v'_{\outputt}) \leq 1$.

\paragraph{Case 2: $\distance_{\tree_S}(c_{\outputt}, c'_{\outputt}) = 1$.}
By the definition of the subdivided clique tree
$\tree_S$, 
exactly one of $c_{\outputt}$ and $c'_{\outputt}$ is a clique vertex, and the other is an intermediate vertex.
Without loss of generality, assume that $c_{\outputt}$ is a clique vertex and $c'_{\outputt}$ is the intermediate vertex subdividing the edge $\{c_{\outputt}, c'\}$ of $\tree$, for some neighbor $c'$ of $c_{\outputt}$.
By the output rule for intermediate vertices, the party corresponding to
$c'_{\outputt}$ outputs the unique vertex $v'_{\outputt} \in V_{\graph}(c_{\outputt}) \cap V_{\graph}(c').$ In particular, $v'_{\outputt} \in V_{\graph}(c_{\outputt})$. Moreover, by the output rule for clique vertices, the party corresponding to
$c_{\outputt}$ outputs a vertex $v_{\outputt} \in V_{\graph}(c_{\outputt}).$
Since $V_{\graph}(c_{\outputt})$ induces a clique in $\graph$, any two vertices in this set are either identical or adjacent. Therefore,
$\distance_{\graph}(v_{\outputt}, v'_{\outputt}) \le 1.$
\end{proof}

The lemma below will enable us to prove the round complexity of $\blockgraphAA(\graph)$.
\begin{lemma} \label{block-graphs:clique-tree-diameter}
 Let $T$ be a minimum-diameter clique tree of a block graph $G$. Then $\diameter(\tree) \le 2\cdot \diameter(\graph)$.
\end{lemma}
\begin{proof}
If $G$ consists of a single maximal clique, then its clique tree $T$ consists of a single vertex, and hence $\diameter(\tree) = 0 \le 2\cdot \diameter(\graph)$. If $G$ contains exactly two maximal cliques, then $T$ consists of a single edge and thus $\diameter(\tree) = 1\le 2\cdot \diameter(\graph)$. For the remainder of the proof, we assume that $G$ contains at least three maximal cliques. In particular, $\diameter(G) \ge 2$. 

Let $\pathh := (c_0, c_1, \ldots, c_k)$ be a longest path in $T$, so that $k = \diameter(\tree)$. If $k \le 2$, then $\diameter(\tree) \le 2 \le 2\cdot \diameter(\graph)$, and the claim follows. Hence, we assume $k \ge 3$.

For every $i \in \{0, 1, \ldots, k - 1\}$, since $c_i$ and $c_{i + 1}$ are adjacent in $\tree$, we have $\vertices_\graph(c_i) \cap \vertices_\graph(c_{i +1}) \neq \emptyset$. As $\graph$ is a block graph, any two distinct maximal cliques intersect in at most one vertex. Therefore, for each $i$, there exists a unique vertex $v_i \in \vertices_\graph(c_i) \cap \vertices_\graph(c_{i+1})$. We then define the vertex sequence $S := (v_0, v_1, \ldots, v_{k-1})$.
We claim that none of the following cases can occur in the sequence $S$:
\begin{enumerate}
    \item There exists an index $i$ such that $\{v_i, v_{i+1}\} \notin E(\graph)$ and $\{v_{i+1}, v_{i+2}\} \notin E(\graph)$.
    \item  
    There exist indices $0 < i + 1 < j \le k - 1$ such that $v_i = v_j$.
\end{enumerate}
We first rule out case~(1).
Since both $v_i$ and $v_{i+1}$ belong to the clique $c_{i+1}$, that is, $\{v_i, v_{i+1}\}\subseteq \vertices_\graph(c_{i+1})$, we have $v_i = v_{i+1}$. This is because, in a clique, every two distinct vertices are adjacent, and if $\{v_i, v_{i+1}\} \notin E(\graph)$, this implies that $v_i = v_{i+1}$. Similarly, since both $v_{i +1}$ and $v_{i+2}$ belong to the clique $c_{i+2}$, i.e., $\{v_i, v_{i+1}\}\subseteq \vertices_\graph(c_{i+2})$, we have $v_{i + 1} = v_{i+2}$. As a result, we obtain $v = v_i = v_{i +1} = v_{i + 2}$. 

Consequently, the cliques $c_i$, $c_{i+1}$, $c_{i + 2}$, and $c_{i+3}$ all intersect in the same vertex $v$, forming a subpath of length three in $T$ whose cliques share a common intersection $\{v\}$. More generally, let $c_{i+1}, \ldots, c_{i+\ell}$ with $\ell > 2$ be a maximal subpath of $\tree$ such that $\cap_{r = i}^{r = i +\ell} \vertices_\graph(c_r) = \{v\}$. Since $k\ge 3$, such a subpath exists with $\ell > 2$.
We then replace this subpath by a star centered at $c_i$, where $c_{i+1}, \ldots, c_{i+\ell}$ become adjacent to $c_i$, while preserving their original neighbors in $\tree$ outside the set  $\{c_{i + 1}, \ldots, c_{i+\ell}\}$. The running intersection property is preserved because all these cliques contain the common vertex $v$, and for any two cliques among $c_i, \ldots, c_{i+\ell}$ the unique clique on the path between them in the modified tree is $c_i$, which also contains $v$. Hence, every clique on the path between any two cliques containing $v$ continues to contain $v$. This transformation
strictly decreases the diameter of $\tree$, contradicting the assumption that $\tree$ is a minimum-diameter clique tree of $\graph$. Hence, $v_i \neq v_{i+1}$ for all $i$ for which $v_{i + 1} = v_{i + 2}$. Thus, case (1) cannot occur.
This implies that in the sequence $S$ no three consecutive vertices can be identical.

We now rule out case (2). Suppose that there exist indices $i$ and $j$ with  $i + 1< j$ such that $v_i = v_j = v$. Then $v \in \vertices_\graph(c_i) \cap \vertices_\graph(c_{i + 1})$ and $v \in \vertices_\graph(c_j) \cap \vertices_\graph(c_{j + 1})$. The running intersection property of clique trees implies that $v$ belongs to every clique on the path from $c_i$ to $c_{j + 1}$ in $\tree$. Let us first consider the case where $v_{i} \neq v_{i + 1}$. Since both $v_i$ and $v_{i + 1}$ belong to $c_{i + 1}$ and $c_{i + 2}$, we have that $\{v_i, v_{i + 1}\} \subseteq  \vertices_\graph(c_{i + 1})\cap \vertices_\graph(c_{i + 2})$. 
This implies that $|\vertices_\graph(c_{i + 1})\cap \vertices_\graph(c_{i + 2})| \ge 2$, which contradicts the assumption that $G$ is a block graph, as in a block graph any two distinct maximal cliques intersect in at most one vertex. 

It remains to consider the case where $v_{i} = v_{i + 1}$. If $v_{i + 1} \neq v_{i + 2}$, then we may follow the same argument as above: replacing $v_i$ by $v_{i +1}$, we obtain that both $v_{i +1}$ and $v_{i +2}$ belong to $c_{i +2}$ and $c_{i +3}$, and hence $\{v_{i +1}, v_{i +2}\} \subseteq \vertices_\graph(c_{i+2}) \cap \vertices_\graph(c_{i+3})$, which again implies that $|\vertices_\graph(c_{i+2}) \cap \vertices_\graph(c_{i+3})| \ge 2$. This contradicts the assumption that $\graph$ is a block graph. On the other hand, if $v_{i + 1} = v_{i + 2}$, then this is exactly case~(1), which we have already shown cannot occur. Therefore, there exist no indices $0 < i + 1 < j \le k - 1$ such that $v_i = v_j$.

Let us construct a sequence $\pathh_\graph$ in $\graph$ using $S$. For every $i \in \{0, 1, \ldots, k - 1\}$, if $v_i \neq v_{i + 1}$, we add $v_i$ to $\pathh_\graph$, and otherwise we skip it. Note that if $v = v_i = v_{i + 1}$ and we skip $v_i$, then by the conclusion of case~(1) $v_{i + 1} \neq v_{i + 2}$, and therefore $v_{i + 1} = v$ will be added to $\pathh_\graph$.  We claim that $\pathh_\graph$ is a path in $\graph$. This is because for every index $i$ with $v_i \neq v_{i + 1}$, the vertices $v_i$ and $v_{i + 1}$ both belong to a maximal clique, and hence they are adjacent in $\graph$, i.e., $\{v_i, v_{i + 1}\} \in E(\graph)$. It follows that consecutive vertices in $\pathh_\graph$ are adjacent in $\graph$. Moreover, by Case~(2), the sequence $S$ contains no two vertices $v_i$ and $v_j$ with $i+1 < j$ such that $v_i = v_j$, and hence the sequence $\pathh_\graph$ contains no repeated vertex. Therefore, $\pathh_\graph$ forms a path in $\graph$. By the construction of $\pathh_\graph$, each vertex of $\pathh_\graph$ can appear at most twice in $S$, and hence $|S| \le 2\cdot |\pathh_\graph|$. In particular, since $|S| = k - 1 = \diameter(\tree) - 1$, we obtain $\diameter(\tree) \le 2\cdot |\pathh_\graph|$.

Since $\pathh_\graph$ is a path in $\graph$, it follows that $|\pathh_\graph| \le \diameter(\graph)$. Combining this with $\diameter(\tree) \le 2\cdot |\pathh_\graph|$, we obtain $\diameter(\tree) \le 2\cdot \diameter(\graph)$.
\end{proof}

We now present the proof of \Cref{theorem:opt-block-graphs}.
\begin{proof}[Proof of \Cref{theorem:opt-block-graphs}]
To compute $\tree_S$, the parties use the deterministic algorithm 
presented in~\cite{CliqueTreesAlg}, which provides them with the same minimum-diameter clique tree $\tree$ of $\graph$.
Afterwards, they execute $\finalTreeAA(\tree_S)$. Then, the parties complete the execution of $\finalTreeAA(\tree_S)$ within $R_{\finalTreeAA}(\diameter(\tree_S))$ rounds. Therefore, they obtain outputs in $\blockgraphAA(\graph)$ within $R_{\finalTreeAA}(\diameter(\tree_S))$ rounds, hence Termination holds.
As $\diameter(\tree_S) \leq 2 \cdot \diameter(\tree)$ by construction, and $\diameter(\tree) \leq 2\cdot \diameter(\graph)$ by \Cref{block-graphs:clique-tree-diameter}, we obtain that $\blockgraphAA(\graph)$ has round complexity $R_{\finalTreeAA}(4 \cdot \diameter(\graph))$, as claimed.
    \Cref{lemma:block-graphs:validity} and \Cref{lemma:block-graphs:1-agreement} guarantee that Validity and $1$-Agreement hold, concluding the proof.
\end{proof}

We obtain our final corollaries for block graphs by instantiating $\finalTreeAA(\tree_S)$ with the protocol described in \Cref{cor:tree-aa} in unauthenticated settings (with no cryptographic assumptions), and with the protocol of \Cref{cor:pki-tree-aa} for authenticated settings (assuming PKI and digital signatures). 
Note that, by \Cref{thm:lower-bound}, we have obtained asymptotically optimal round complexity whenever $t \in \Theta(n)$.

\begin{corollary}
    There is a synchronous protocol that achieves $\approximateagreement$ on any connected block graph $\graph$ in $O\left(\tfrac{\log \diameter(\graph)}{\log\log \diameter(\graph)}\right)$ rounds given that up to $t < n/3$ of the $n$ parties involved are Byzantine.
\end{corollary}

\begin{corollary}
    There is a synchronous protocol that achieves $\approximateagreement$ on any connected block graph $\graph$ in $O\left(\tfrac{\log \diameter(\graph)}{\log\log \diameter(\graph)}\right)$ rounds in authenticated settings given that up to $t = c \cdot n$ of the $n$ parties are Byzantine, where $c$ is a constant in $[1/3, 1/2)$.
\end{corollary}

Moreover, we note that $\blockgraphAA(\graph)$ does not depend on the network being synchronous -- it achieves $\approximateagreement$ on $\graph$ whenever the underlying $\finalTreeAA(\tree_S)$ protocol achieves $\approximateagreement$ on the subdivided clique tree $\tree_S$. By instantiating $\finalTreeAA(\tree_S)$ with the asynchronous protocol of \cite[Theorem 5]{DISC:NoRy19}, we obtain another corollary, for the asynchronous model. We recall that the asynchronous model only assumes that messages get delivered eventually.

\begin{corollary}
    There is an asynchronous protocol that solves $\approximateagreement$ on any connected block graph $\graph$ given that up to $t < n/3$ of the $n$ parties involved are Byzantine.
\end{corollary}
\section{Conclusions}

In this work, we have investigated the optimal round complexity for achieving $\approximateagreement$ on trees in the synchronous model.
Our results extend previous findings from real-valued domains to tree-structured input spaces, providing novel insights for achieving efficient $\approximateagreement$ in discrete input spaces.
We presented a protocol with round complexity $O\left( \frac{\log \diameter(\tree)}{\log \log \diameter(\tree)} \right)$, improving upon the previously best-known $O(\log \diameter(\tree))$-round protocol of \cite{DISC:NoRy19}.  Our construction achieves optimal resilience and relies on a reduction from $\approximateagreement$ on trees to its real-valued counterpart.
Moreover, we have provided an asymptotically matching lower bound
by adapting the lower bound of \cite{Fekete90} from real values to graphs, demonstrating that our protocol achieves asymptotically optimal round complexity whenever $t \in \Theta(n)$. 

We further extend our techniques to block graphs by exploiting their clique tree structure. By reducing $\approximateagreement$ on a block graph to $\approximateagreement$ on a suitably constructed tree, we obtain efficient synchronous and asynchronous protocols with optimal resilience, with asymptotically optimal round complexity in the synchronous model whenever $t \in \Theta(n)$.

In contrast to prior work on $\approximateagreement$ outside the real line, our protocols rely only on efficient local computation and avoid the exponential-time safe-area constructions.

Our results suggest several directions for further investigation. On the practical side, improving the constants in the round complexity (particularly those inherited from real-valued $\approximateagreement$) would directly improve the efficiency of our protocols for trees and block graphs. On the theoretical side, it remains open whether similarly optimal-resilience efficient protocols can be obtained for broader classes of graphs, and where the precise structural boundary lies between $\approximateagreement$ and $\convexagreement$ in terms of resilience. More broadly, we view reduction-based techniques such as ours as a promising approach for efficient $\approximateagreement$ across diverse input spaces.

\newpage

\newpage
\bibliographystyle{plain}
\bibliography{bib/abbrev3,bib/crypto,bib/project}

\newpage
\appendix
\section*{Appendix}

\section{Protocol for trees: Additional Proofs} \label{appendix:realvalues-remarks}

We include the proof of the second part of \cref{remark:combined}. 
\begin{proof}
    Assume without loss of generality that $j < j'$. Since $j' - j \le 1$, either there is an integer $z$ such that $z \leq j, j', \leq z + 1$ or there is an integer $z$ such that $z - 1 \leq j \leq z \leq j' \leq z + 1$. In the first case, both $\closestInt{j}$ and $\closestInt{j'}$ are either $z$ or $z + 1$ by the definition of $\closestInt{\cdot}$. In the second case, assume that $\closestInt{j} = z - 1$, and $\closestInt{j'} = z + 1$. Then, $j - (z - 1) < z - j$, hence $j < z - \frac{1}{2}$,
    while $(z + 1) - j' \leq j' - z$, and therefore $j' \geq z + \frac{1}{2}$. This leads to a contradiction, as it implies $j' - j > 1$. Hence, in both cases, $\closestInt{j'} - \closestInt{j} \leq 1$.
\end{proof}

\section{$\realAA$ in the $t < n/3$ setting}\label{appendix:realvalues-aa}
As mentioned in Section \ref{section:warmup}, the analysis of \cite{BenDoHo10} regarding $\realAA$ proves that $\approximateagreement$ is achieved for $\varepsilon = 1/n$. Theorem \ref{theorem:real-values-aa} extends the analysis for any $\varepsilon > 0$.

In order to prove Theorem \ref{theorem:real-values-aa}, we make use of a few results in the full version \cite{BDH10} of \cite{BenDoHo10}.
We first need to note the round complexity of each iteration in $\realAA$. This remark follows directly from \cite[Theorem~1]{BenDoHo10}, which discusses the guarantees of the mechanism the parties use to distribute their current values in each iteration.
\begin{remark}[Theorem 1 of \cite{BDH10}]\label{remark:rounds-per-iteration}
    Each iteration of $\realAA(\varepsilon)$ takes three rounds.
\end{remark}

Our proof also utilizes the two claims stated below: these are concerned with $\varepsilon$-Agreement and Validity. In the following, $V_{R}$ denotes the multiset of values held by the honest parties at the end of iteration $R$ (hence, round $3 \cdot R$), and $V_0$ denotes the multiset of honest inputs.

\begin{lemma}[Claim 12 of \cite{BDH10}]
    \label{lmm:Claim12}
    In every iteration $R$, $(\max V_R - \min V_{R}) \leq (\max V_0 - \min V_0) \cdot \frac{t^R}{R^R \cdot (n - 2t)^R}$. 
\end{lemma}

\begin{lemma}[Claim 8 of \cite{BDH10}] \label{lemma:claim8}
   In every iteration $R$, $V_R \subseteq [\min V_0, \max V_0]$.
\end{lemma}

We are now ready to prove Theorem \ref{theorem:real-values-aa}, which we restate below.

\RealValuesAA*
\begin{proof}
    By the Lemma \ref{lmm:Claim12}, we know that the maximum absolute difference between the output values of any two honest parties is at most 
    $
        \frac{\diameter \cdot t^R}{(n-2t)^R \cdot R^R}
    $
    after $R$ iterations. Since $n > 3t$, this is at most $\frac{\diameter}{R^R}$. We will now show that for $R := \lceil \frac{20}{9} \cdot \frac{\log \delta}{\log \log \delta} \rceil$, where $\delta := \diameter/\varepsilon$, we obtain 
    $\frac{\diameter}{R^R} \leq \varepsilon$
    and thus $\varepsilon$-Agreement holds after at most $R$ iterations. Below we provide a lower bound for $\log_2(R^R)$:    
    \begin{align*}
        \log_2(R^R) &= R \log_2 R \\
        &\geq \frac{20}{9} \cdot \frac{\log_2 \delta}{\log_2 \log_2 \delta} \cdot (\log_2 \log_2 \delta - \log_2 \log_2 \log_2 \delta) \\
        &\geq \frac{20}{9} \cdot\frac{\log_2 \delta}{\log_2 \log_2 \delta} \cdot (\log_2 \log_2 \delta - \frac{11}{20} \cdot \log_2 \log_2 \delta) \\
        &=\log_2 \delta \\
        &= \log_2\left( \frac{\diameter}{\varepsilon} \right).
    \end{align*}
    By exponentiation on both sides (with base $2$), we obtain that $R^R \geq \frac{\diameter}{\varepsilon}$, and therefore our claim $\frac{\diameter}{R^R} \leq \varepsilon$ holds, which proves that $\varepsilon$-Agreement is reached.
    Note that, in the third line, we use the fact that since the function $\log_2(x)/x$, defined for $x > 0$, has a global maximum (at $x = e$) of $\frac{1}{e \cdot \ln(2)} < 11/20$ and therefore also $\log_2(\log_2\log_2(x))/\log_2(\log_2(x)) < 11/20$ for all $x > 0$.

    As each of the iterations of $\realAA$ requires three rounds according to \Cref{remark:rounds-per-iteration}, the protocol terminates after at most
    $
         R_{\realAA}(\diameter, \varepsilon) \leq 3 \cdot R < 7 \cdot \frac{\log_2 (\diameter/\varepsilon)}{\log_2 \log_2 (\diameter/\varepsilon)} + 3
    $
    rounds. In addition, Validity holds due to Lemma \ref{lemma:claim8}, hence $\approximateagreement$ is achieved.
\end{proof}

Below we present the proof of \Cref{remark:convex-hull-maintained}.
\ObviousHullMaintained*
\begin{proof}
    Let $v$ be an arbitrary vertex in $S$. If $v \in \vertices(\pathh)$, then $v = \projection_{\pathh}(v)$ and the statement follows immediately.
    The remainder of the proof is therefore concerned with the case where $v \notin \vertices(\pathh)$.
    
    Let $v_h$ be a vertex in $\vertices(\pathh) \cap \hull{S}$. Note that all vertices on the path $\pathh(v, v_h)$ are in $\hull{S}$ since both $v$ and $v_h$ are in $\hull{S}$.
    We label the vertices in $\pathh(v, v_h)$ as $(v_1 := v, v_2, \ldots, v_m := v_h)$.
    As we know that $v_1$ is not on $\pathh$ but $v_h$ is on $\pathh$, there is an $i$ satisfying $1 < i \leq h$ such that $v_i \in \vertices(\pathh)$ and $v_1, \ldots, v_{i-1} \not \in \vertices(\pathh)$.

    Then, $v_i$ must be the projection of $v$ onto $\pathh$, i.e., $v_i = \projection_{\pathh}(v)$. Assuming otherwise implies that either there is another vertex within $v_1, \ldots, v_{i-1}$ that is in $\vertices(\pathh)$, which contradicts the way we have defined $i$, or that the tree $\tree$ contains a cycle. Therefore, because $v_i \in \vertices(\pathh)$ is the projection of $v$ onto $\pathh$, it follows by our previous observation $v_i \in \hull{S}$ that $\projection_{\pathh}(v) \in \vertices(\pathh) \cap \hull{S}$.
\end{proof}

\section{$\realAA$ in the $t < n/2$ setting} \label{appendix:real-values-pki}

In this section, we focus on $\approximateagreement$ in the authenticated setting, where we can handle up to $t < n / 2$ Byzantine parties. \Cref{theorem:real-values-aa} already handles the range $t < n / 3$, hence in the following we only focus on the remaining range, i.e., $t$ satisfying $n / 3 \leq t < n / 2$. 
We discuss the proof of \Cref{theorem:realAA-pki}, restated below.
\RealAAPKI*

As mentioned in the main body of the paper, we obtain this result by combining two protocols: the Proxcensus protocol of \cite{EUROCRYPT:GhGoLi22}, which will be used when $D / \varepsilon$ is large, and the $\approximateagreement$ protocol of \cite{PODC:GhLiWa22}, which will be used otherwise.

\paragraph{When $D/\varepsilon$ is large.}
As previously mentioned, we rely on the protocol of \cite{EUROCRYPT:GhGoLi22}. This protocol implements, in fact, \emph{Binary Proxcensus} -- the parties join with inputs in $\{0, 1\}$. For a given integer $\ell \geq 1$, the parties output $1$-close integer values in $\{0, 1, \ldots, \ell\}$ such that: if all honest parties have input $0$, all honest parties output $0$, and if all honest parties have input $1$, all honest parties output $1$. In the protocol of \cite{EUROCRYPT:GhGoLi22}, the parties first \emph{expand} their inputs from $\{0, 1\}$ to $\{0, \ell\}$: input $0$ remains $0$, and input $1$ becomes $\ell$. Afterwards, the parties run a series of iterations where they compute new values with the guarantee that (i) honest parties' new values are in the range of the honest parties' values at the beginning of the iteration, (ii) the range of honest values is reduced with each iteration. We note that the new values are computed as a \emph{rounded} average (so that parties hold integer values in each iteration), but removing the rounding maintains this guarantee.
Then, this is essentially an $\approximateagreement$ protocol on $[0, \ell]$: the only adjustments that we need to make are (i) allowing parties to join directly with inputs from a bounded interval instead of mapping binary inputs to $\{0, \ell\}$, (ii) removing the rounding in the new values' computation, and (iii) fixing the number of iterations based on our use-case. We will see that, if honest parties hold $\diameter$-close inputs and we want their outputs to be $\varepsilon$-close, $O\left( \frac{\log (D/\varepsilon)}{\log \log (D / \varepsilon)} \right)$ iterations suffice whenever $t = c \cdot n$ for some constant $c \in [1/3, 1/2)$. More precisely, $\lceil \frac{20}{9 \cdot c'} \cdot \frac{\log_2 \delta}{\log_2 \log_2 \delta} \rceil$ iterations, where $c'  :=  \frac{1 - 2c}{c}$ and $\delta := \big(\diameter/\varepsilon\big)^{c'}$, will suffice. In the following, we refer to the adjusted protocol as $\realAA^\pki_1$, and we prove the following statement.
\begin{theorem}\label{thm:proxcensus-to-AA}
 If the honest inputs are $\diameter$-close real values, there is a protocol $\realAA^\pki_1(\diameter, \varepsilon)$ achieving $\approximateagreement$ even when up to $t = c \cdot n$ (where $c$ is a constant in $[1/3, 1/2)$) of the $n$ parties are Byzantine, ensuring termination within  
    \[R_{\realAA^\pki_1}(\diameter, \varepsilon)
    < 
    7 \cdot\frac{\log_2(\diameter/\varepsilon)}{\log_2 \big((1 - 2c)/c\big)+ \log_2\log_2(\diameter/\varepsilon)} + 3
    \]
    rounds. 
\end{theorem}

Note that, $\log_2 \big((1 - 2c)/c\big)$ is non-positive for our chosen range of $c \in [1/3, 1/2)$. However, the denominator is non-positive only when $\diameter / \varepsilon$ is a constant -- in this case, $\realAA^\pki$ will make use of the protocol of \cite{PODC:GhLiWa22} instead: this protocol achieves Termination within $O(\diameter / \varepsilon) = O(1)$ rounds.

We first note that each iteration requires three rounds of communication, as stated in \cite{EUROCRYPT:GhGoLi22}.
\begin{remark}\label{remark:rounds-per-iteration-new}
    Each iteration of $\realAA^\pki_1(\varepsilon)$ takes three rounds.
\end{remark}

Then $V_R$ denotes the multiset of values obtained by the honest parties in iteration $R$, and $V_0$ denotes the multiset of honest inputs. 
The lemma below establishes that Validity is maintained at all times.
\begin{lemma}[Lemma 4 of \cite{EUROCRYPT:GhGoLi22}] \label{lemma:claim8-new}
    After $R$ iterations, $V_R \subseteq [\min V_0, \max V_0]$.
\end{lemma}

The following lemma follows from Lemma 6 of \cite{EUROCRYPT:GhGoLi22}. We note that this lemma highlights the need for the constant $c < 1/2$, as setting $n := 2t + 1$ instead would mean that the honest values are not guaranteed to converge unless $R > t$.
\begin{lemma}[Lemma 6 in \cite{EUROCRYPT:GhGoLi22}]
    \label{lmm:Claim12-new}
    After $R$ iterations, $\max V_R - \min V_{R} \leq (\max V_0 - \min V_0) \cdot \frac{t^R}{R^R \cdot (n - 2t)^R}$. 
\end{lemma}

The lemma below discusses the number of iterations that are sufficient for $\varepsilon$-Agreement to be achieved. Recall that $t = c \cdot n$, where $c \in [1/3, 1/2)$ is a constant, and an upper bound on the honest inputs' range size $\diameter$, i.e., $\diameter \geq \max V_0 - \min V_0$, is assumed to be known. The proof will be similar to that of $\realAA$ presented in the proof of \Cref{theorem:real-values-aa}.
\begin{lemma} \label{lemma:sufficient-iterations-pki}
    Let  $R := \lceil \frac{20}{9 \cdot c'} \cdot \frac{\log_2 \delta}{\log_2 \log_2 \delta} \rceil$, where $\delta := \big(\diameter/\varepsilon\big)^{c'}$ and $c' := \frac{1 - 2c}{c}$. Then, 
    \begin{align*}
        \max V_{R} - \min V_{R} \leq \varepsilon.
    \end{align*}
\end{lemma}
\begin{proof}
    By \Cref{lmm:Claim12-new} we have the upper bound upper bound $\max V_{R} - \min V_{R } \leq \diameter \cdot \frac{t^R}{(n-2t)^R \cdot R^R}$. As $t = c \cdot n$, we can make the following simplification: $
    \frac{t^R}{(n-2t)^R}
    \leq \frac{(c \cdot n)^R}{\left( (1 - 2c) \cdot n \right)^R} = \frac{c^R}{(1 - 2c)^R}$.
    
    Let $c' :=  \frac{1 - 2c}{c}$.  This yields that $\frac{\diameter \cdot t^R}{(n-2t)^R \cdot R^R}$ is at most $\frac{\diameter}{c'^R R^R}$. We claim that choosing  $R := \lceil \frac{20}{9 \cdot c'} \cdot \frac{\log_2 \delta}{\log_2 \log_2 \delta} \rceil$, where $\delta := \big(\diameter/\varepsilon\big)^{c'}$, leads to $\frac{\diameter}{c'^R R^R} \leq \varepsilon$
    and thus $\varepsilon$-Agreement holds after at most $R$ iterations. Below we provide a lower bound for $\log_2\left((c'\cdot R)^R \right)$. 
    \begin{align*}
        \log_2\left((c'\cdot R)^R \right) &= R \log_2 (c' \cdot R) \\
        &\geq \frac{1}{c'} \cdot \frac{20}{9} \cdot \frac{\log_2 \delta}{\log_2 \log_2 \delta} \cdot \log_2\left( \frac{20}{9} \cdot \frac{\log_2 \delta}{\log_2 \log_2 \delta} \right) \\
        &\ge \frac{1}{c'} \cdot \frac{20}{9} \cdot \frac{\log_2 \delta}{\log_2 \log_2 \delta} \cdot (\log_2 \log_2 \delta - \log_2 \log_2 \log_2 \delta)\\
        &\geq \frac{1}{c'} \cdot\frac{20}{9} \cdot\frac{\log_2 \delta}{\log_2 \log_2 \delta} \cdot (\log_2 \log_2 \delta - \frac{11}{20} \cdot \log_2 \log_2 \delta) \\
        & \ge \frac{1}{c'}\cdot \log_2 \delta\\
        &= \log_2\left( \frac{\diameter}{\varepsilon} \right).
    \end{align*}
    By exponentiation on both sides (with base $2$), we obtain that $(c' \cdot R)^R \geq \frac{\diameter}{\varepsilon}$, and therefore our claim follows by the following lines: 
    \begin{align*}
        \max V_{R} - \min V_{ R} \leq \diameter \cdot \frac{t^R}{(n-2t)^R \cdot R^R} \leq \frac{\diameter}{c'^R \cdot R^R} \leq \frac{\varepsilon \cdot c'^R \cdot R^R}{c'^R \cdot R^R} = \varepsilon
    \end{align*}
    The other simplifications used here are the same as we used in \Cref{theorem:real-values-aa}.
\end{proof}

We may now present the proof of \Cref{thm:proxcensus-to-AA}. The round complexity analysis will be similar to that of $\realAA$ presented in the proof of \Cref{theorem:real-values-aa}.
\begin{proof}[Proof of \Cref{thm:proxcensus-to-AA}]
    As $\realAA^\pki_1$ runs for a fixed number of iterations $R := \lceil \frac{20}{9 \cdot c'} \cdot \frac{\log_2 \delta}{\log_2 \log_2 \delta} \rceil$, where $\delta := \big(\diameter/\varepsilon\big)^{c'}$ and $c' := \frac{1 - 2c}{c}$, Termination is guaranteed. Since each of the iterations of $\realAA^\pki_1$ requires three rounds according to \Cref{remark:rounds-per-iteration-new}, the protocol terminates after at most
    \begin{align*}
        R_{\realAA^\pki_1}(\diameter, \varepsilon) \leq 3 R <  7 \cdot \frac{\log_2 (\diameter/\varepsilon)}{\log_2 c' + \log_2 \log_2 (\diameter/\varepsilon)} + 3
    \end{align*}
        rounds. As $c=O(1)$, also $c'$ is a constant and hence the asymptotic runtime is as described in the theorem's statement. In addition, \Cref{lemma:sufficient-iterations-pki} ensures that $\varepsilon$-Agreement holds, and Validity holds due to Lemma \ref{lemma:claim8-new}, hence $\approximateagreement$ is achieved.
\end{proof}

\paragraph{When $D/\varepsilon$ is small.} If $\log_2 \big((1-2c)/c\big)+ \log_2\log_2(\diameter/\varepsilon)  \leq 0$ (and hence $\diameter / \varepsilon$ is a constant), we make use of the $\approximateagreement$ protocol of \cite{PODC:GhLiWa22}, which we denote by $\realAA^\pki_2$. This protocol is designed in the \emph{network-agnostic model}: given $t_s, t_a$ such that $t_a \leq t_s$ and $2 \cdot t_s + t_a < n$, $\realAA^\pki_2$ tolerates up to $t_s$ Byzantine corruptions when running in a synchronous network, and up to $t_a$ corruptions when running in an asynchronous network. By setting $t_a := 0$, this becomes a synchronous $\approximateagreement$ protocol resilient against $t = t_s < n / 2$ Byzantine corruptions. $\realAA^\pki_2$ follows the outline described in our paper's introduction, halving the honest values' range with each iteration. Each iteration takes four rounds of communication, hence $\approximateagreement$ is achieved within $O(\log (D / \varepsilon))$ rounds, which in this case is a constant number of rounds.
\begin{theorem}[Theorem 3.1 of \cite{PODC:GhLiWa22}]\label{thm:pki-AA-pretty}
 Given that the honest inputs are $\diameter$-close real values, there is a protocol $\realAA^\pki_2(\diameter, \varepsilon)$ achieving $\approximateagreement$ even when $t < n/2$ of the $n$ parties are Byzantine, ensuring Termination within $R_{\realAA^\pki_2}(\diameter, \varepsilon) = 4 \cdot  \lceil \log_2 (\diameter / \varepsilon) \rceil$ rounds. 
\end{theorem}

\paragraph{Putting it all together.}
We may now provide the proof of \Cref{theorem:realAA-pki}, where we describe the protocol $\realAA^\pki$.
\begin{proof}[Proof of \Cref{theorem:realAA-pki}]
    In $\realAA^\pki(\diameter, \varepsilon)$, the parties first check locally whether the condition $\log_2 \big((1 - 2c)/c\big)+ \log_2\log_2(\diameter/\varepsilon) > 0$ holds. If this is the case, the parties run $\realAA^\pki_1(\diameter, \varepsilon)$, described by \Cref{thm:proxcensus-to-AA}. Otherwise, the parties run $\realAA^\pki_2(\diameter, \varepsilon)$, described by \Cref{thm:pki-AA-pretty}. $\approximateagreement$ is achieved either way.
    
    Regarding round complexity, note that, if $\log_2 \big((1 - 2c)/c\big)+ \log_2\log_2(\diameter/\varepsilon) \leq 0$, then $\diameter/\varepsilon$ is a constant. This implies that $\realAA^\pki_2(\diameter, \varepsilon)$ according to \Cref{thm:pki-AA-pretty}, and hence $\realAA^\pki(\diameter, \varepsilon)$, terminates within $O(1)$ rounds. Otherwise, if $\log_2 \big((1-2c)/c\big)+ \log_2\log_2(\diameter/\varepsilon) > 0$, \Cref{thm:proxcensus-to-AA} ensures that $\realAA^\pki_1(\diameter, \varepsilon)$, and therefore $\realAA^\pki(\diameter, \varepsilon)$ terminates within $O\left( \frac{\log_2(\diameter/\varepsilon)}{\log_2\log_2(\diameter/\varepsilon)}
    \right)$ rounds.
\end{proof}

\section{Proof of \Cref{lemma:block-graphs:clique-separator}}\label{sec:appndix_proof_of_lemma}

In this section, we present the proof of \Cref{lemma:block-graphs:clique-separator}. For completeness, we first restate the lemma and then provide its proof.

\CUTVERTEX*

\begin{proof}
We first note that $\{u\} = \vertices_{\graph}(c)\cap \vertices_{\graph}(c')$: $c$ and $c'$ represent maximal cliques by the definition of clique trees and, by the definition of block graphs, every two maximal cliques intersect in at most one vertex. Assume for contradiction that there exists an $x$--$y$ path
$P := (v_0=x, v_1, \dots, v_m=y)$ in $G$ such that $u \notin \{v_0,\dots,v_m\}$. Without loss of generality, assume that $P$ is a shortest such path.

For each $i \in \{0,\dots,m-1\}$, let $c_i$ be a maximal clique of $G$ that contains the edge $\{v_i,v_{i+1}\}$.\footnote{Such a maximal clique exists since every edge of $G$ is contained in a block, and every block is a clique in a block graph.}
Then for every $i\in\{1,\dots,m-1\}$, $v_i \in \vertices_{\graph}(c_{i-1}) \cap \vertices_{\graph}(c_i)$, and hence $\vertices_{\graph}(c_{i-1}) \cap \vertices_{\graph}(c_i) \neq \emptyset$.
By the running intersection property of clique trees, every clique on the unique path in $\tree$ between $c_{i - 1} $ and $c_i$ also contains vertex $v_i$. Therefore, the union of all these clique-tree paths (over all indices $i$) forms a connected subtree $\tree_P$ of $\tree$ that contains all cliques $c_0$, $c_1$, $\dots$, $c_{m - 1}$. 

Since $x \in V_c$ and $y \in V_{c'}$, the subtree $\tree_P$ intersects both components $T_c$ and $T_{c'}$ of $T \setminus \{\{c,c'\}\}$. Therefore, $\tree_P$ must contain two adjacent clique vertices $d \in\vertices(T_c)$ and $d' \in\vertices(T_{c'})$ joined by the edge $\{c,c'\}$.

Since $d$ and $d'$ are adjacent in the clique tree, their corresponding cliques intersect.
By the running intersection property, we have that
$\vertices_{\graph}(d) \cap \vertices_{\graph}(d') \subseteq \vertices_{\graph}(c) \cap \vertices_{\graph}(c') = \{u\}.$ Hence $\vertices_{\graph}(d) \cap \vertices_{\graph}(d') = \{u\}$.

Because both $d$ and $d'$ belong to $\tree_P$, there exists a vertex $v_i$ of the path $P$
contained in both $\vertices_{\graph}(d)$ and $\vertices_{\graph}(d')$.
Thus $v_i = u$, contradicting the assumption that $P$ avoids $u$.
\end{proof}

\section{Analysis of $\listConstruction$} \label{appendix:trees}

We present the proof of Lemma \ref{lemma:list-construction}, describing the guarantees of algorithm $\listConstruction$. The parties run this algorithm locally to convert the input tree $T$ into a list representation that provides a few special properties.
\ListConstruction*
\begin{proof}
First, we prove that $\listConstruction(\tree, \roott)$ terminates in a finite amount of time. The algorithm follows a depth-first search ($\dfs$) traversal, which has a time complexity of $O(|V| + |E|)$ for a graph with vertex set $V$ and edge set $E$. Since $\tree$ is a tree, it has $|E| = |V| - 1$ edges, leading to an overall complexity of $O(|V|)$. Thus, the algorithm terminates within a finite number of steps. Then, we prove each property in order.

\paragraph{Property 1.} This follows directly from the definition of $\dfs$. At each step of the $\dfs$ traversal, if a vertex has an unvisited child, we move to that child. Otherwise, we backtrack to the parent. This ensures that each successive pair of vertices in $L$ represents either a move to a child or a return to a parent, both of which are adjacent in $T$.

\paragraph{Property 2.} We first show that $\abs{L} \leq 2 \cdot \abs{\vertices(\tree)}$ holds. We proceed by induction on the number of vertices $m := |\vertices(\tree)|$. 

\textbf{Base case:} If $\tree$ consists of a single vertex, then $L = [\roott]$. Since $\abs{L} = 1 \leq 2 \cdot |\vertices(\tree)| = 2$, the claim holds.

\textbf{Inductive step:} Suppose the claim holds for all trees with at most $m-1$ vertices. Consider a tree $\tree$ with $m$ vertices. Let $\tree'$ be the tree obtained by removing a leaf $v$ and its incident edge $(u,v)$.

By the induction hypothesis, the $\dfs$ traversal on $\tree'$ results in a list $L'$ of length at most $2(m-1)$. When we add $v$ back, the $\dfs$ first visits $v$ once and then backtracks to $u$, adding two elements to $L'$. Therefore, $|L| = |L'| + 2 \leq 2(m-1) + 2 = 2m$, which completes the induction.

To show that every vertex appears in $L$, note that the $\dfs$ visits each vertex at least once upon first exploration, ensuring that every vertex is included.

\paragraph{Property 3.} During $\dfs$, the first occurrence of $v$ in $L$ happens before visiting any of its descendants. The last occurrence happens after all descendants have been visited. Hence, the indices corresponding to $v$ enclose all its descendants in $L$. If $u$ is in the subtree of $v$, all its occurrences in $L$ appear between $i_{\min}$ and $i_{\max}$.
Conversely, if $L(u) \subseteq [i_{\min}, i_{\max}]$, then $\dfs$ must have discovered $u$ while visiting $v$'s subtree, confirming that $u$ is in the subtree rooted at $v$.

\paragraph{Property 4.} If one of $v$ or $v'$ is an ancestor of the other, then the lowest common ancestor ($\lca$) of two nodes $v$ and $v'$ is trivially in this range.

Otherwise, let $w := \lca(v, v')$. If $w \notin {L_k : \min(i, i') \leq k \leq \max(i, i')}$, then $\dfs$ would have visited $v$'s subtree and then $v'$'s subtree without encountering $w$. However, this contradicts the $\dfs$ property that it does not leave a subtree before visiting all its nodes. Thus, $w$ must appear in the interval, completing the proof.
\end{proof}

\end{document}